\documentclass[11pt]{article}
\usepackage{graphicx}
\usepackage[cmex10]{amsmath}
\usepackage{algorithmic}
\usepackage{array}
\usepackage[caption=false,font=normalsize,labelfont=sf,textfont=sf]{subfig}
\usepackage{fixltx2e}
\usepackage{stfloats}
\usepackage{url}
\usepackage{amssymb}
\usepackage{subfig}
\usepackage{diagbox}
\usepackage{hyperref}
\usepackage{color}
\usepackage{flushend}
\newtheorem{theorem}{Theorem}

\newtheorem{corollary}{Corollary}
\newtheorem{remark}{Remark}

\newcommand{\dff}{\stackrel{\scriptscriptstyle\triangle}{=}}

\begin{document}
\title{Multiaccess Communication via a Broadcast Approach\\ 
Adapted to the Multiuser Channel}

\author{Samia Kazemi and  Ali Tajer\thanks{Authors are with the Electrical, Computer, and Systems Engineering Department, Rensselaer Polytechnic Institute, Troy, NY 12180.}
}
\date{}

\maketitle
\allowdisplaybreaks

\begin{abstract}
A broadcast strategy for multiple access communication over slowly fading channels is introduced, in which the channel state information is known to only the receiver. In this strategy, the transmitters split their information streams into multiple independent information streams, each adapted to a specific actual channel realization. The major distinction between the proposed strategy and the existing ones is that in the existing approaches, each transmitter adapts its transmission strategy only to the fading process of its direct channel to the receiver, hence directly adopting a single-user strategy previously designed for the single-user channels. However, the contribution of each user to a network-wide measure (e.g., sum-rate capacity) depends not only on the user's direct channel to the receiver, but also on the qualities of other channels. Driven by this premise, this paper proposes an alternative broadcast strategy in which the transmitters adapt their transmissions to the combined states resulting from all users' channels. This leads to generating a larger number of information streams by each transmitter and adopting a different decoding strategy by the receiver. An achievable rate region and an outer bound that capture the trade-off among the rates of different information layers are established, and it is shown that the achievable rate region subsumes the existing known capacity regions obtained based on adapting the broadcast approach to the single-user channels. 
\end{abstract}

\section{Introduction}
\label{sec:intro}

Random fluctuations of the wireless channel states induce uncertainty about the network state at all transmitter and receiver sites \cite{Biglieri}. Slowly varying channels can be estimated by the receivers with high fidelity, rendering the availability of the channel state information (CSI) at the receiver. Acquiring the CSI by the transmitters can be further facilitated via feedback from the receivers, which incurs additional communication and delay costs. The instantaneous and ergodic performance limits of the multiple access channel (MAC) with the CSI available to all transmitters and the receiver are well-investigated~\cite{Biglieri,ahlsewde,liao}. In certain communication scenarios, however, acquiring the CSI by the transmitters is not viable due to, e.g., stringent delay constraints or excessive feedback costs. In such scenarios, the notion of outage capacity evaluates the likelihood for the reliable communication for a fixed transmission rate~\cite{Ozarow}. When the actual channel realization can sustain the rate, transmission is carried out successfully, and otherwise, it fails and no message is decoded~\cite{Biglieri,Ozarow}. The notions of outage and delay-limited capacities are studied extensively for various networks including the multiple access channel (c.f.~\cite{HanlyTse,LiJindalGoldsmith,narasimhan,Haghi,DasNarayan,jafar} and references therein).

Superposition coding is shown to be an effective approach for circumventing CSI uncertainty at the transmitters. The underlying motivation for this approach is that each transmitter splits its data stream into a number of independently generated coded streams with possibly different rates. These streams are superimposed and transmitted by the designated transmitter, and the receiver decodes as many streams as the quality of the channel affords. The aggregate rate of transmission, subsequently, is the sum of individual rates of the streams decoded by the receiver. Motivated by superposition coding, and following the broadcast approach to compound channels~\cite{Cover:IT72}, the notion of broadcast strategy for slowly fading single-user channel was initially introduced for effective single-user communication~\cite{Shamai:ISIT97}. In this approach, any channel realization is viewed as a broadcast receiver, rendering an equivalent network consisting of a number of receivers. Each receiver is designated to a specific channel realization and is degraded with respect to a subset of other channels. The broadcast strategy is further generalized to single-user channels with mixed delay constraints in \cite{kfirSteinerShamai}, and single-user multi-antenna channels \cite{Shamai:IT03}, where the singular values of channel matrices are leveraged to rank and order the degradedness of different channel realizations. 

The effectiveness of broadcast strategy for multiuser channels is investigated in~\cite{ShamaiMac} and \cite{Minero:ISIT07} for the settings in which the transmitters have uncertainties about all channels, and in~\cite{Shamai:ISIT2013} for the settings in which each transmitter has uncertainties about the channels of other users. Specifically, the approaches in~\cite{ShamaiMac} and \cite{Minero:ISIT07} adopt the broadcast strategy designed for single-user channels, and directly apply it to the MAC. As a result, each transmitter generates a number of information streams, each adapted to a specific realization of the direct channel linking the transmitter to the receiver. An alternative scenario in which each transmitter has the CSI of its direct channel to the receiver while being unaware of the states of other users' channels is studied in~\cite{Shamai:ISIT2013}, where a transmission approach based on rate splitting and sequential decoding is proposed.

In this paper, we take a different approach based on the premise that the contribution of each user to the overall performance of the multiple access channel not only depends on the direct channel linking this user to the receiver, but also is influenced by the {\em relative} qualities of the other users' channels. Hence, we propose a strategy in which the information streams are generated and adapted to the combined state of the channel resulting from incorporating all individual channel states. In order to highlight the distinction with the existing approaches, consider a two-user MAC in which each channel takes one of the two possible states, referred to as {\em weak} and {\em strong} channels. The approach of~\cite{Minero:ISIT07} assigns two streams to each transmitter, one apt for the weak channel, and the second one suited to the strong channel. Each transmitter generates and transmits these streams without regards for the possible states of the other user's channel. In the proposed approach, in contrast, we leverage the fact that the two channels take a combination of four possible states. Hence, every transmitter generates four information streams, each suited to one of the four possible states. The proposed approach leads to an equivalent network with a number of receivers each corresponding to one possible combination of all channels. We provide an achievable rate region and an outer bound on the capacity of this resulting multi-terminal network, and show that the achievable rate region of this network is considerably larger than the capacity region of the model presented in~\cite{Minero:ISIT07}. The proposed approach is further extended from the two-state channel to the general finite-state channels, and the corresponding achievable rate region is characterized.  We remark that the  discrepancy and improvement in the capacity region compared to~\cite{Minero:ISIT07} is due modeling the channel differently, which facilitates a finer resolution in adapting the codebooks to the channel states as well as in decoding them. 

The remainder of the paper is organized as follows. The finite-state channel model is presented in Section~\ref{sec:model}. The layering (rate-splitting) and the successive decoding approach, which constitute the proposed broadcast approach, as well as an achievable rate region and an outer bound for the two-state channel are presented in Section~\ref{sec:two}. The proposed approach and the achievable rate region are generalized to the finite-state channel setting in Section~\ref{sec:finite}, and Section~\ref{sec:conclusions} concludes the paper.

\vspace{-.05 in}
\section{Channel Model}
\label{sec:model}

Consider a two-user multiple access channel, in which  two independent users transmit independent messages to a common receiver via a discrete-time Gaussian multiple-access fading channel.  All the users are equipped with one antenna and the random channel coefficients  independently take one of the $\ell\in\mathbb{N}$ distinct values, denoted by $\{\sqrt\alpha_m: m\in\{1,\dots,\ell\}\}$. The fading process is assumed to remain unchanged during each transmission cycle, and can change to independent states afterwards. Channel states are {\em unknown} to  transmitters, while the receiver is assumed to have full CSI. The users are subject to an average transmission power constraint $P$. By defining $X_i$ as the signal of transmitter $i\in\{1,2\}$ and $h_i$ as the coefficient of the channel linking transmitter $i\in\{1,2\}$ to the receiver, the received signal is
\begin{equation} \label{eq:model}
Y = h_1X_1 + h_2X_2 + N\ ,
\end{equation}
where $N$ accounts for the additive white Gaussian with mean zero and variance 1. Depending on the realization of the channels $h_1$ and $h_2$, the multiple access channel can be in one of the $\ell^2$ possible states. By leveraging the broadcast approach (c.f. \cite{Shamai:ISIT97, Shamai:IT03}, and \cite{Minero:ISIT07}), the communication model in \eqref{eq:model} can be equivalently presented by a broadcast network that has two inputs $X_1$ and $X_2$ and $\ell^2$ outputs. Each output corresponds to one possible combinations of channels $h_1$ and $h_2$. We denote the output corresponding to the combination $h_1=\sqrt{\alpha_m}$ and $h_2=\sqrt{\alpha_n}$ by
\begin{equation} \label{eq:model2}
Y_{mn} = \sqrt{\alpha_m} X_1 + \sqrt{\alpha_n} X_2 + N_{mn}\ ,
\end{equation}
where $N_{mn}$ is a standard Gaussian random variable for all $m,n\in \{1,\dots,\ell\}$. Figure~\ref{fig:network} depicts this network for the case of the two-state channels ($\ell=2$). Without loss of generality and for the convenience in notations, we assume that channel gains $\{\alpha_m: m\in\{1,\dots,\ell\}\}$ take real positive values and are ordered in the ascending order, i.e., 
\begin{align}
0<\alpha_1< \alpha_2< \dots < \alpha_\ell\ .
\end{align}
We use the notation $C(x,y)\dff \frac{1}{2}\log_2\Big(1+\frac{x}{y+\frac{1}{P}}\big)$ throughout the paper.
\begin{figure}[t]
\centering
\includegraphics[width=2.2in]{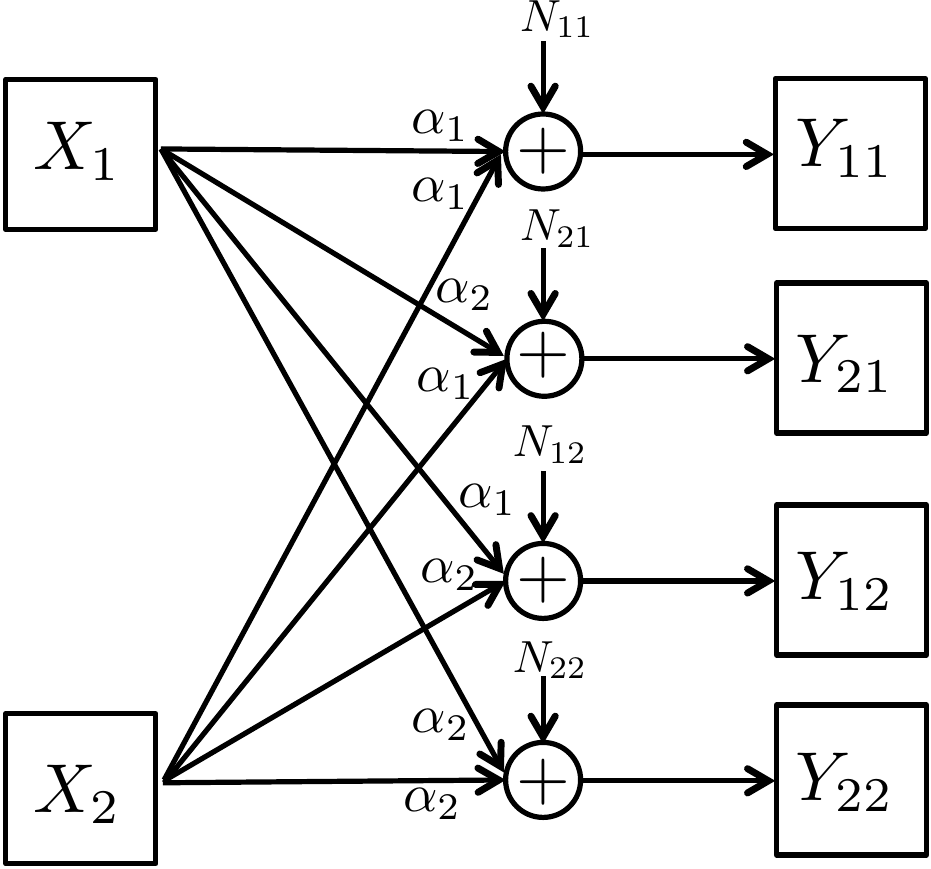}
\caption{Equivalent degraded broadcast channel corresponding to a two user four state multiple access channel with channel gains $\alpha_1$ and $\alpha_2$.}
\label{fig:network}
\end{figure}

\vspace{-.1 in}
\section{Two-state Channels ($\ell=2$)}
\label{sec:two}

We start by analyzing the setting in which the channels take one of the two possible values, i.e., $\ell=2$. This setting furnishes the context in order to highlight the differences between the proposed streaming and successive decoding strategy in this paper and those investigated in \cite{Minero:ISIT07}. By leveraging the intuition gained from the two-state setting, we generalize the codebook generation and the successive decoding strategies to accommodate a fading process with any arbitrary number of finite channel states in Section~\ref{sec:finite}. Throughout the rest of this section, we refer to channels $\alpha_1$ and $\alpha_2$ as the {\em weak} and {\em strong} channels, respectively.

%\vspace{-.1 in}
%\newpage
\subsection{Background: Adapting streams to the Single-user Channels}

In order to motivate the proposed approach, we start by reviewing the broadcast strategy concept for a single-user channel introduced in~\cite{Shamai:ISIT97}, and its generalization for the two-user multiple access channel investigated in~\cite{Minero:ISIT07}. When facing a two-state channel, the single-user strategy of~\cite{Shamai:ISIT97} splits the information stream of the transmitter into two streams, each corresponding to one fading state, and encodes them independently. The two encoded information streams are subsequently superimposed and transmitted over the channel. One of the streams,  denoted by $W_1$, is always decoded by the receiver, while the second stream, denoted by $W_2$, is decoded only when the channel is {\em strong}. The successive decoding order adopted in this approach is presented in Table~\ref{table:1}.

{\footnotesize
\begin{table}[!h]
\renewcommand{\arraystretch}{1.3}
\caption{Successive decoding order of~\cite{Shamai:IT03}}
\label{table:1}
\centering
\begin{tabular}{ |c||c|c|c|c| } 
 \hline
 $h^2$ & Decoding stage 1 & Decoding stage 2\\ 
 \hline\hline
 $\alpha_1$ & $W_{1}$ &\\ 
 \hline
 $\alpha_2$ & $W_{1}$ & $W_{2}$\\
 \hline
\end{tabular}
\end{table}
}
This strategy is adopted and directly applied to the multiple access channel in~\cite{Minero:ISIT07}. Specifically, it generates two coded information streams per transmitter, where the streams of user $i\in\{1,2\}$ are denoted by $\{W^i_1,W^i_2\}$. Based on the actual realizations of the channels, a combination of these streams are successively decoded by the receiver. 
%Specifically, similar to the single-user case, successive decoding consists of up to two stage. 
In the first stage, the baseline streams $W^1_1$ and $W^2_1$, which constitute the minimum amount of guaranteed information, are decoded. Additionally, when the channel between transmitter $i$ and the receiver, i.e., $h_i$ is strong, in the second stage information stream $W^i_2$ is also decoded. Table~\ref{table:tse} depicts the decoding sequence corresponding to each of the four possible channel combinations.

\begin{table}[!h]
\renewcommand{\arraystretch}{1.3}
\caption{Successive decoding order of~\cite{Minero:ISIT07}}
\label{table:tse}
\centering
\begin{tabular}{ |c||c|c|c|c| } 
 \hline
 $(h_1^2,h_2^2)$ & Decoding stage 1 & Decoding stage 2\\ 
 \hline\hline
 $(\alpha_1,\alpha_1)$ & $W^1_1,W^2_1$ &\\ 
 \hline
 $(\alpha_2,\alpha_1)$ & $W^1_1,W^2_1$ & $W^1_2$\\ 
 \hline
 $(\alpha_1,\alpha_2)$ & $W^1_1,W^2_1$ & $W^2_2$\\ 
 \hline
 $(\alpha_2,\alpha_2)$ & $W^1_1,W^2_1$ & $W^1_2,W^2_2$\\ 
 \hline
\end{tabular}
\end{table}

\vspace{-.1 in}
\subsection{Adapting streams to the MAC}
Contribution of user $i\in\{1,2\}$ to a network-wide performance metric (e.g., sum-rate capacity) depends not only on the quality of the channel $h_i$, but also on the quality of the channel of the other user. This motivates assigning more information streams to user $i$ and adapting them to the {\em combined} effect of {\em both} channels, instead of adapting them only to channel $h_i$. Designing and assigning more than two information streams to each transmitter facilitates a finer resolution in successive decoding, which in turn expands the capacity region characterized in~\cite{Minero:ISIT07}.

%In order to enable adapting the transmission of each user to both channels, we 
We assume that each transmitter splits its message into {\em four} streams corresponding to the four possible combinations of  the two channels. These codebooks for transmitter $i\in\{1,2\}$ are denoted by $\{W^i_{11},W^i_{12},W^i_{21},W^i_{22}\}$, where the information stream $W^i_{uv}$ is associated with the channel realization in which the channel gain of user $i$ is $\alpha_v$, and the channel gain of the other user is $\alpha_u$. These stream assignments are demonstrated in Fig.~\ref{fig:code}.

The initial streams $\{W^1_{11},W^2_{11}\}$ account for the minimum amount of guaranteed information, which are adapted to the channel combination $(h^2_1,h^2_2)=(\alpha_1,\alpha_1)$ and should be decoded by all four possible channel combinations. When at least one of the channels is strong, the remaining codebooks are grouped and adapted to different channel realizations according to the assignments described in Fig.~\ref{fig:code}. Specifically: 
\begin{itemize}
\item The second group of streams $\{W^1_{12},W^2_{21}\}$ are reserved to be decoded in addition to  $\{W^1_{11},W^2_{11}\}$ when $h_1$ is strong, while $h_2$ is still weak. 
\item Alternatively, when $h_1$ is weak and $h_2$ is strong, instead the third group of streams, i.e., $\{W^1_{21},W^2_{12}\}$, are decoded.
\item Finally, when both channels are strong, in addition to all the previous streams, the fourth group $\{W^1_{22},W^2_{22}\}$ is also decoded. 
\end{itemize}\vspace{-.1 in}

\begin{figure}[!h]
\centering
\includegraphics[width=3.5in]{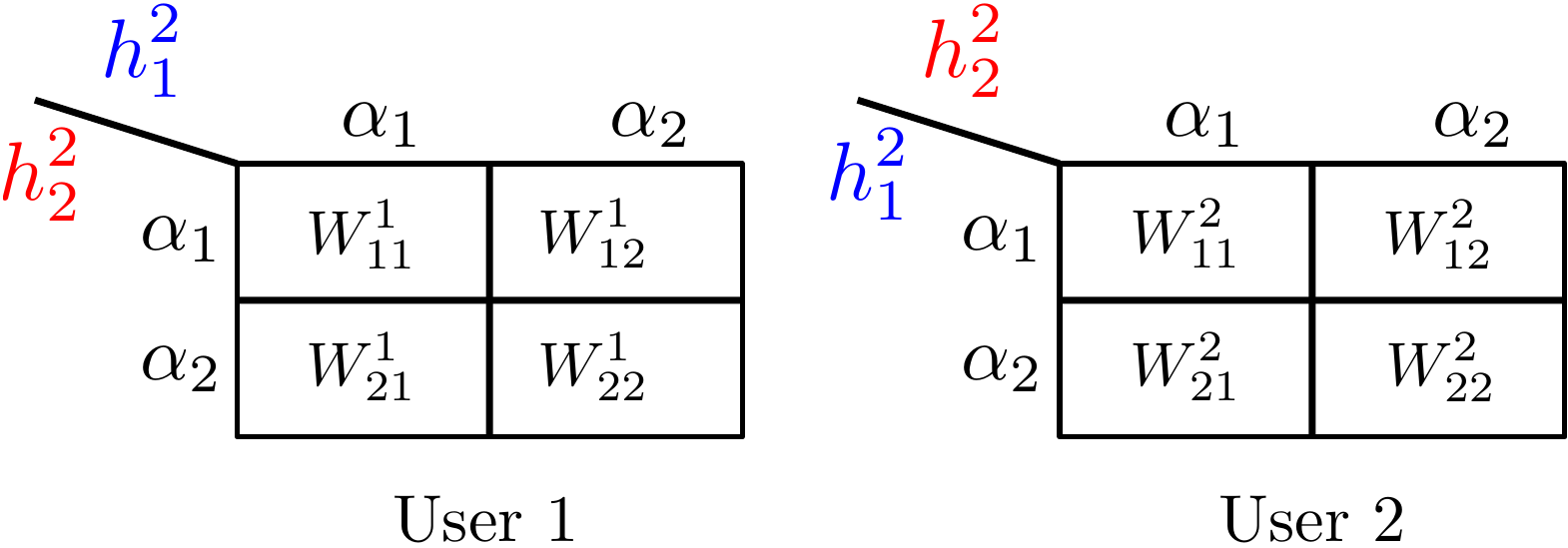}
\caption{Streaming and codebook assignments by user 1 and user 2.}
\label{fig:code} \vspace{-.1 in}
\end{figure}

\begin{table}[!b]
\renewcommand{\arraystretch}{1.3}
\caption{Successive decoding order of the streams adapted to the MAC}
\label{table:MAC}
\centering
\begin{tabular}{ |c||c|c|c|c| } 
 \hline
 $(h^2_1,h^2_2)$ & stage 1 & stage 2 & stage 3\\ 
 \hline\hline
 $(\alpha_1,\alpha_1)$ & $W^1_{11},W^2_{11}$ & & \\ 
 \hline
 $(\alpha_2,\alpha_1)$ & $W^1_{11},W^2_{11}$ & $W^1_{12}, W^2_{21}$ & \\ 
 \hline
 $(\alpha_1,\alpha_2)$ & $W^1_{11},W^2_{11}$ & $W^1_{21}, W^2_{12}$ & \\ 
 \hline
 $(\alpha_2,\alpha_2)$ & $W^1_{11},W^2_{11}$ & $W^1_{12}, W^2_{12},W^1_{21}, W^2_{21}$ & $W^1_{22},W^2_{22}$\\ 
 \hline
\end{tabular}
\end{table}

The orders of successive decoding for different combinations of channel realizations are presented in Table~\ref{table:MAC}.  Based on this successive decoding order, channel gain state $(\alpha_1, \alpha_1)$ is degraded with respect to all other states (i.e., the capacity region of the MAC corresponding to receiver $Y_{11}$ is strictly smaller than those of the other three receivers), while $(\alpha_1, \alpha_2)$ and $(\alpha_2,\alpha_1)$ are degraded with respect to $(\alpha_2,\alpha_2)$.  Clearly, the codebook assignment and successive decoding approach presented in Table~\ref{table:MAC} subsumes the one proposed in~\cite{Minero:ISIT07}, as presented in Table~\ref{table:tse}. In particular, Table~\ref{table:tse} can be recovered as a special case of Table~\ref{table:MAC} by setting the rates of the streams $\{W^1_{21},W^2_{21},W^1_{22},W^2_{22}\}$ to zero.  This implies that the proposed strategy should perform no worse than the one described in Table~\ref{table:tse}. This codebook assignment and decoding order gives rise to the equivalent broadcast network with two inputs $\{X_1,X_2\}$ and four outputs $\{Y_{11},Y_{12},Y_{21},Y_{22}\}$, as depicted in Fig.~\ref{fig:network2}.
\begin{remark}
Adopting the proposed broadcast approach transforms the original two-user MAC to a multi-terminal network consisting of two transmitters and four receivers, where each receiver is designated to decode a pre-specified set of codebooks. The resulting multi-terminal network model is {\em different} from that of~\cite{Minero:ISIT07}, and is expected to have a different capacity region. In the following two subsections, we provide an achievable rate region and an outer bound on the capacity of the network depicted in Fig.~\ref{fig:network2}. 
\end{remark}

\begin{remark}
For adopting the notion of  {\em broadcast approach} to the settings beyond the single-user single-antenna settings, the key element to be borrowed and generalized is the concept of degradedness, which allows for ordering the channels based on their qualities. In multi-user settings, this notion is not always as well-defined as in the single-user single-antenna case, and often involves heuristic ways of ordering the channels. In the approach proposed in this subsection, we use the capacity region of the multiple access channels formed from the transmitters to each of the four possible receivers, where it can be readily verified that as for each of these four multiple access channels, the capacity region expands as one of the channels becomes stronger.
\end{remark}

\begin{figure}[t]
\centering
\includegraphics[width=4in]{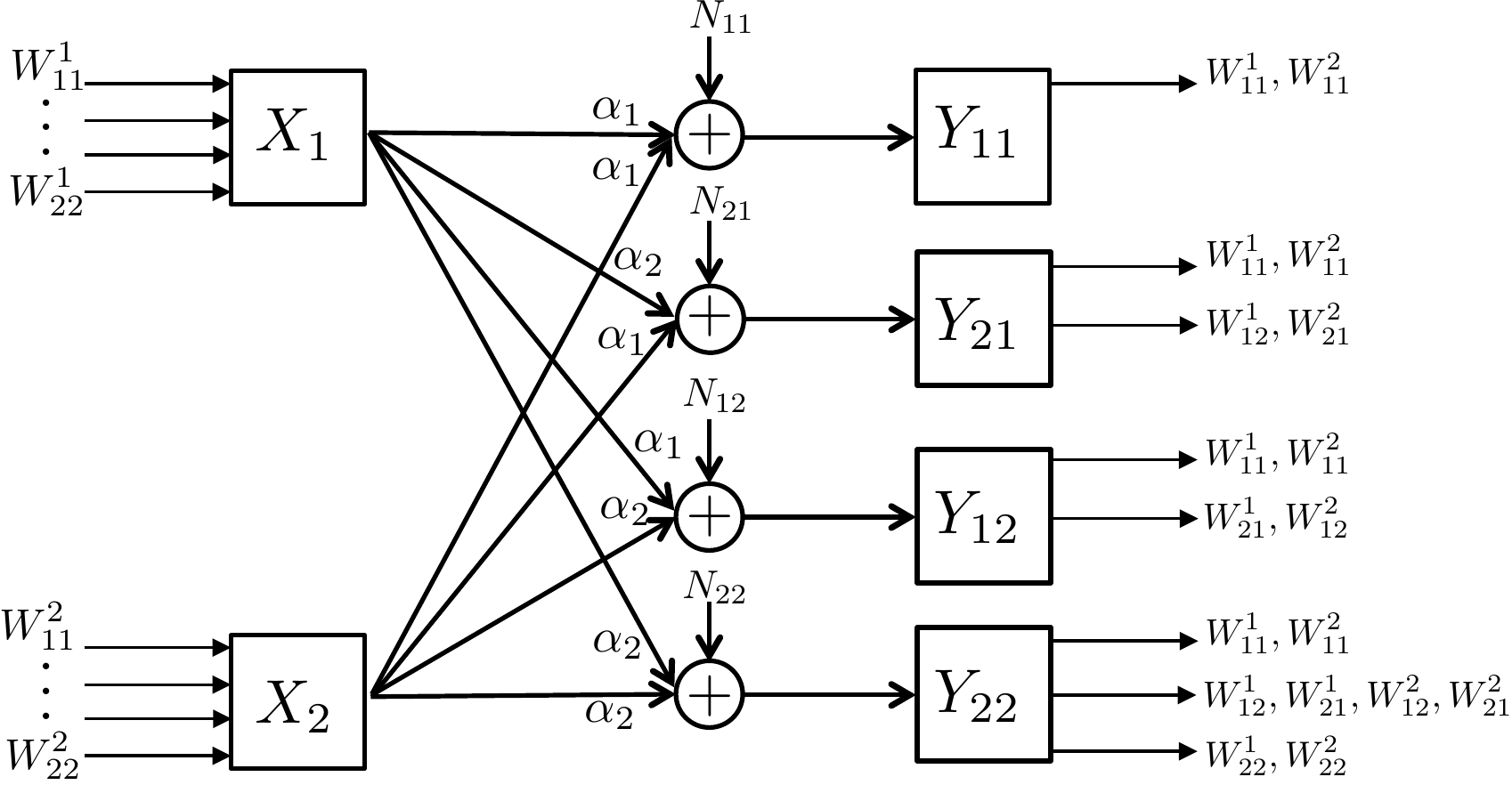}
\caption{Equivalent network with two inputs and four outputs.}
\label{fig:network2}
\end{figure}

\subsection{Achievable Rate Region}
This subsection delineates the region of all achievable rates $R^i_{uv}$ for  $i,u,v\in\{1,2\}$, where $R^i_{uv}$ accounts for the rate of codebook $W^i_{uv}$. We define $\beta^i_{uv}\in[0,1]$ as the fraction of the power that transmitter $i$ allocates to streams $W^i_{uv}$ for $u\in\{1,2\}$ and $v\in\{1,2\}$, where we clearly have $\sum_{u=1}^2\sum_{v=1}^2\beta^i_{uv}=1$. For the convenience in notations, and in order to place the emphasis on the interplay among the rates of different information streams, we consider the case that relevant streams in different users have identical rates, i.e., rates of information streams $W^1_{uv}$ and $W^2_{uv}$, denoted by $R^1_{uv}$ and $R^2_{uv}$ respectively, are the same, and denoted by $R_{uv}$, i.e., $R_{uv}\dff R^1_{uv}=R^2_{uv}$.  The results can be readily generalized to arbitrarily different rates for different streams.

\begin{theorem}[Achievable Region]
\label{theorem:achievable_rate2}
The achievable rate region of the rates $(R_{11}, R_{12}, R_{21}, R_{22})$ for the channel depicted in Fig.~\ref{fig:network2} is the set of all rates satisfying:
\begin{align}
\label{R:11} R_{11} \; & \; \leq \; r_{11}\\
\label{R:12} R_{12} \; & \; \leq \; r_{12} \\
\label{R:21} R_{21} \; & \; \leq \;  r_{21}\\
\label{R:1221} R_{12}+R_{21} \; &\leq  \;  r_1 \\
\label{R:21221} 2R_{12}+R_{21} \; & \leq \;  r_{12}'\\
\label{R:12221} R_{12}+2R_{21} \; & \;  \leq \; r_{21}'\\
\label{R:22} R_{22} \; & \; \leq r_{22}\ ,
 \end{align}
 where $\{r_{11},r_{12},r_{21},r_{22},r_1,r_{12}',r_{21}'\}$ are defined in Appendix~\ref{appendix:a}, over all possible power allocation factors  $\beta^i_{uv}\in[0,1]$ such that $\Sigma_{u=1}^2\Sigma_{v=1}^2\beta^i_{uv}=1$.
\end{theorem}
\begin{proof}
The proof follows from the structure of the rate-splitting approach presented in Fig.~\ref{fig:code} and the decoding strategy presented in Table~\ref{table:MAC}. Detailed proof is provided in Appendix \ref{appendix:theorem:achievable_rate2}.
\end{proof}
In order to compare the achievable rate region in Theorem~\ref{theorem:achievable_rate2} and the capacity region presented in~\cite{Minero:ISIT07}, we group the information streams in the way that they are ordered and decoded in~\cite{Minero:ISIT07}. Specifically, the streams $\{W^1_{21},W^2_{21},W^1_{22},W^2_{22}\}$ are allocated zero power. Information streams $W_{11}^1$ and $W_{11}^2$ are adapted to the weak channels, and the information streams $W_{12}^2$ and $W_{12}^2$ are reserved to be decoded when one or both channels are strong. Information streams adapted to the {strong} channels are grouped and their rates are aggregated, and those adapted to the weak channels are also groups and their rates are aggregated. Based on this, the region presented in Theorem~\ref{theorem:achievable_rate2} can be used to form the sum-rates $R_w\dff (R^1_{11}+R^2_{11})$ and $R_s\dff (R^1_{12}+R^2_{12})$. 
 
\begin{corollary}\label{corollary:1}
By setting the power allocated to streams $\{W^1_{21},W^2_{21},W^1_{22},W^2_{22}\}$ to zero, the achievable rate region characterized by Theorem~\ref{theorem:achievable_rate2} reduces to the following region, which coincides with the capacity region characterized in \cite{Minero:ISIT07}.
\begin{align}
\label{eq:Rw} R_w\;&\; \leq \min\{a_3,a_6,a_9,a_4+a_8\}\ ,\\
\label{eq:Rs} \mbox{and}\quad R_s \;&\; \leq  C\left(  {\alpha_2 \beta^1_{12} +\alpha_2 \beta^2_{12} \; , \; 0} \right)\ ,
\end{align}
where $\{a_3,a_4,a_6,a_8,a_9\}$ are defined in Appendix~\ref{appendix:theorem:achievable_rate2}.
\end{corollary}
\begin{proof}
See Appendix \ref{appendix:corollary:1}.  
\end{proof}

\subsection{Outer Bound}
In this subsection, we present an outer bound for the capacity region of the network in Fig.~\ref{fig:network2} for our proposed encoding and decoding strategy.
\begin{theorem}[Outer Bound]
\label{theorem:outer_bound2}
An outer bound for the capacity region of the rates $(R_{11}, R_{12}, R_{21}, R_{22})$ for the channel depicted in Fig.~\ref{fig:network2} is the set of all rates satisfying:
\begin{align}
R_{11} \; & \; \leq \; \frac{1}{2}a_{3}, \quad R_{12} \; \leq \; \frac{1}{2}a_{24}\ ,\quad R_{21}  \; \leq \;  \frac{1}{2}a_{27}, \quad R_{22} \; \leq r_{22}\ ,
\end{align}
where $r_{22}$ is defined in Appendix~\ref{appendix:a} and   constants $\{a_{24},a_{27}\}$ are defined in Appendix~\ref{appendix:theorem:achievable_rate2}.
\end{theorem}
\begin{proof}
See Appendix \ref{app:outer}.
\end{proof}
Distance of this outer bound from the achievable rate region depends on the values of the channel parameters, i.e., channel coefficients, transmission power, and noise variance. 

\begin{figure}[t]
\centering
\includegraphics[width=3.3in]{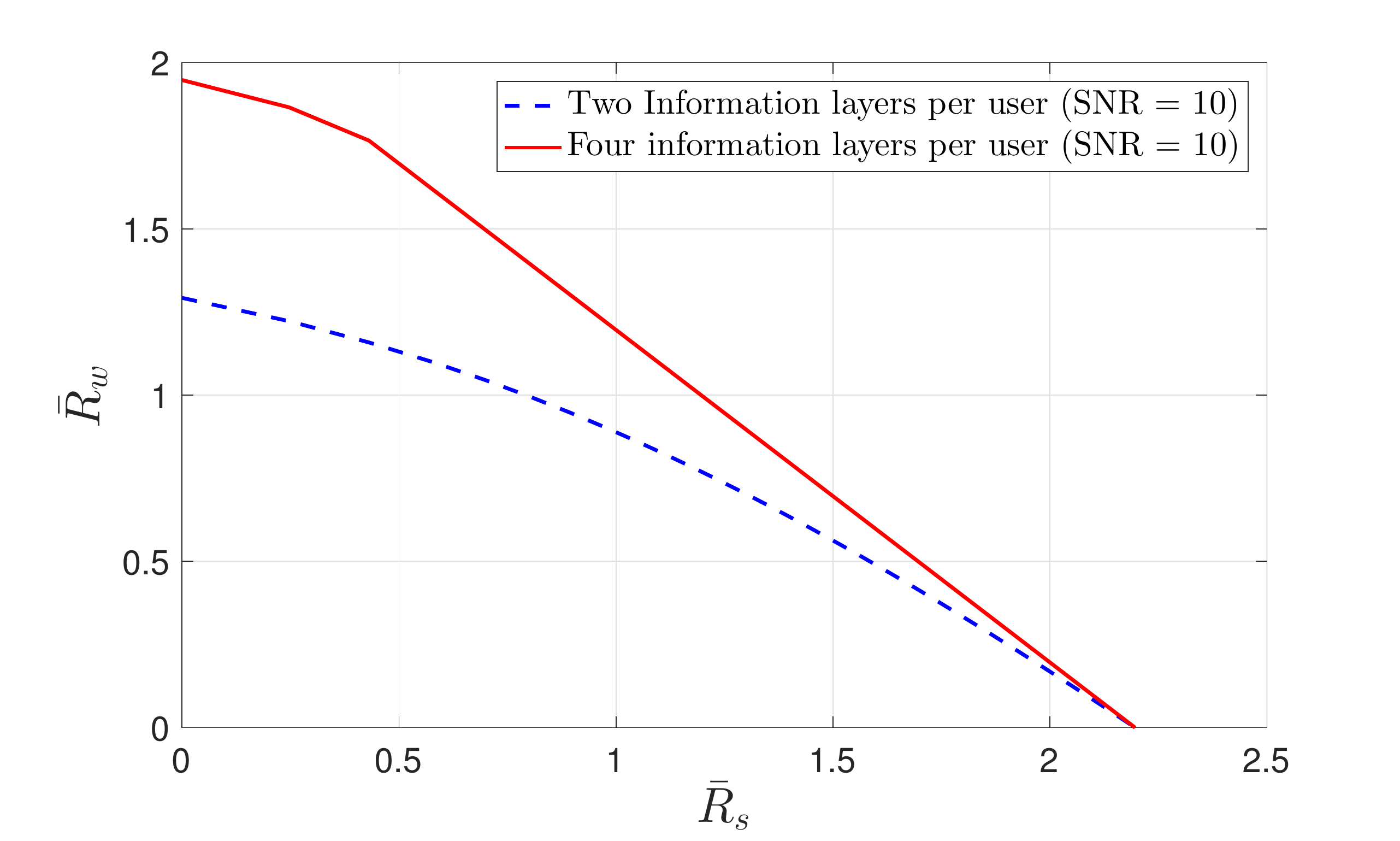}
\caption{Comparison of the capacity region presented in~\cite{Minero:ISIT07} and achievable rate region in Theorem~\ref{theorem:achievable_rate2} demonstrating the trade-off between $R_s$ and $R_w$, and $\bar R_s$ and $\bar R_w$. Here, transmission signal-to-noise ratio (SNR) is 10, the channel
coefficients are $(\sqrt{\alpha_1}, \sqrt{\alpha_2}) = (0.5, 1)$.}
\label{fig:region_snr_10}
\end{figure}

\subsection{Numerical Evaluations}
\label{subsec:Numerical}
First, we assess and compare the achievable rate region for the proposed approach in this paper (Theorem~\ref{theorem:achievable_rate2}) with the region provided by Corollary~\ref{corollary:1} and \cite{Minero:ISIT07} in Fig.~\ref{fig:region_snr_10}. Since the latter ones evaluate the trade-off between the sum-rates of the information streams adapted to the weak and strong channels, we provide the comparison as the same trade-off. For this purpose, corresponding to the coding scheme of \cite{Minero:ISIT07} (Table~\ref{table:tse}) we have earlier defined the sum-rates
\begin{align}
\label{eq:Rsw2}
R_w = R^1_{11}+R^2_{11} \ ,\quad  \mbox{and}\quad R_s  = R^1_{12}+R^2_{12}\ ,
\end{align}
and for the coding scheme proposed in this paper (Table~\ref{table:MAC}) we define
\begin{align}
\label{eq:Rsw4}
\bar R_w \dff R^1_{11}+R^2_{11}+R_{21}^1+R_{21}^2+R_{12}^1+R_{12}^2\ , \quad  \mbox{and}\quad \bar R_s  \dff R^1_{22}+R^2_{22}\ .
\end{align}
Figure~\ref{fig:region_snr_10} demonstrates the regions described by $(R_w,R_s)$ and $(\bar R_w,\bar R_s)$, in which the transmission signal-to-noise ratio (SNR) is 10, the channel coefficients are $(\sqrt{\alpha_1},\sqrt{\alpha_2})=(0.5,1)$, and the regions are optimized over all possible power allocation ratios.  The  numerical evaluation in Fig.~\ref{fig:region_snr_10} confirms that the achievable rate region in Theorem~\ref{theorem:achievable_rate2} dominates that of Corollary~\ref{corollary:1}, and the gap between two regions diminishes as the rates of the information layers adapted to the strong channels increases, i.e., $R_s$ and $\bar R_s$ increase.

\begin{figure}[t]
\centering
\begin{minipage}{.5\textwidth}
  \centering
  \includegraphics[width=3.3 in]{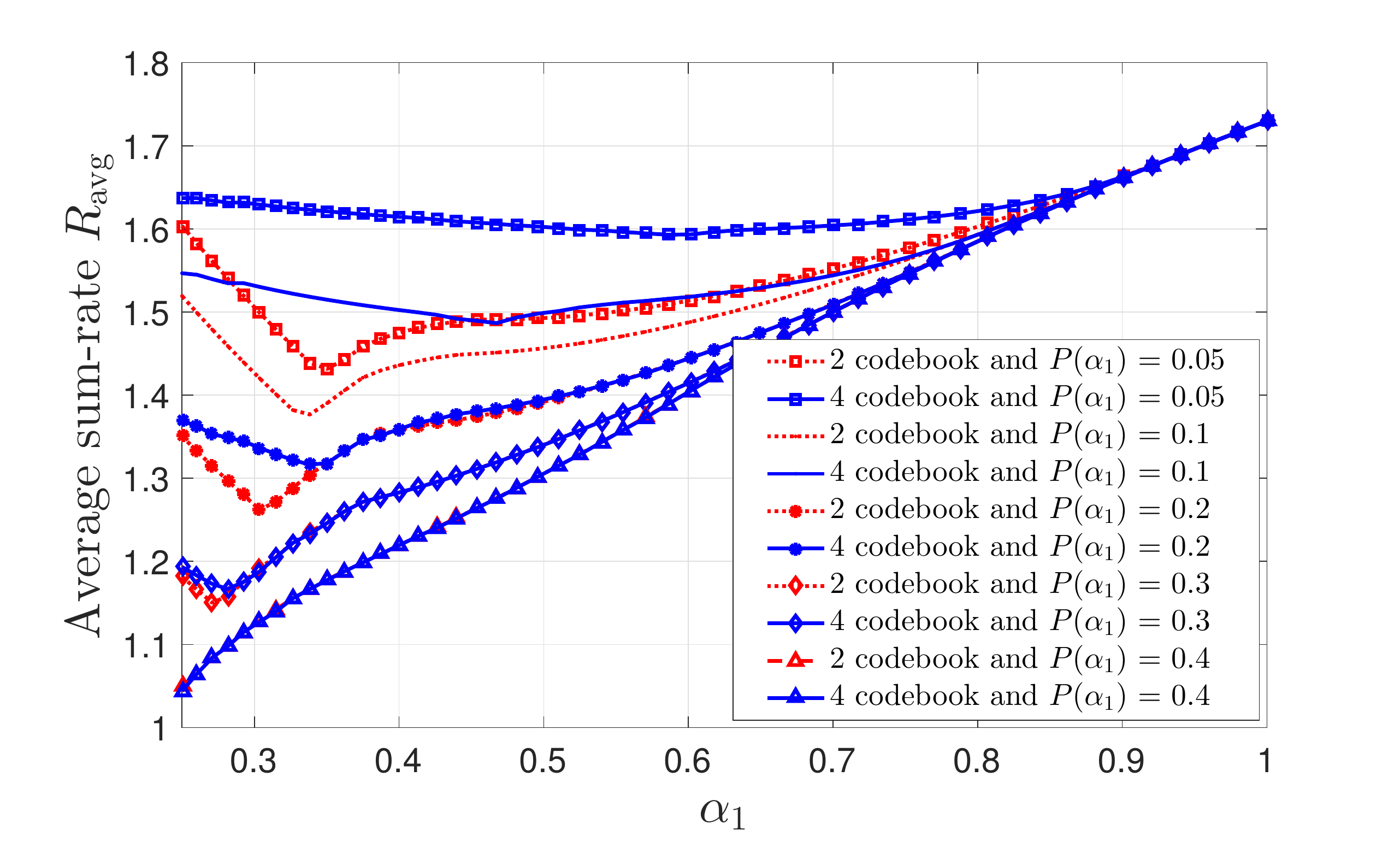}
  \caption{Average sum-rate versus $\alpha_1$ for different values of $p$ ($\alpha_2=1$ and SNR=5).}
  \label{fig:rate_alpha}
\end{minipage}%
\begin{minipage}{.5\textwidth}
  \centering
  \includegraphics[width=3.3 in]{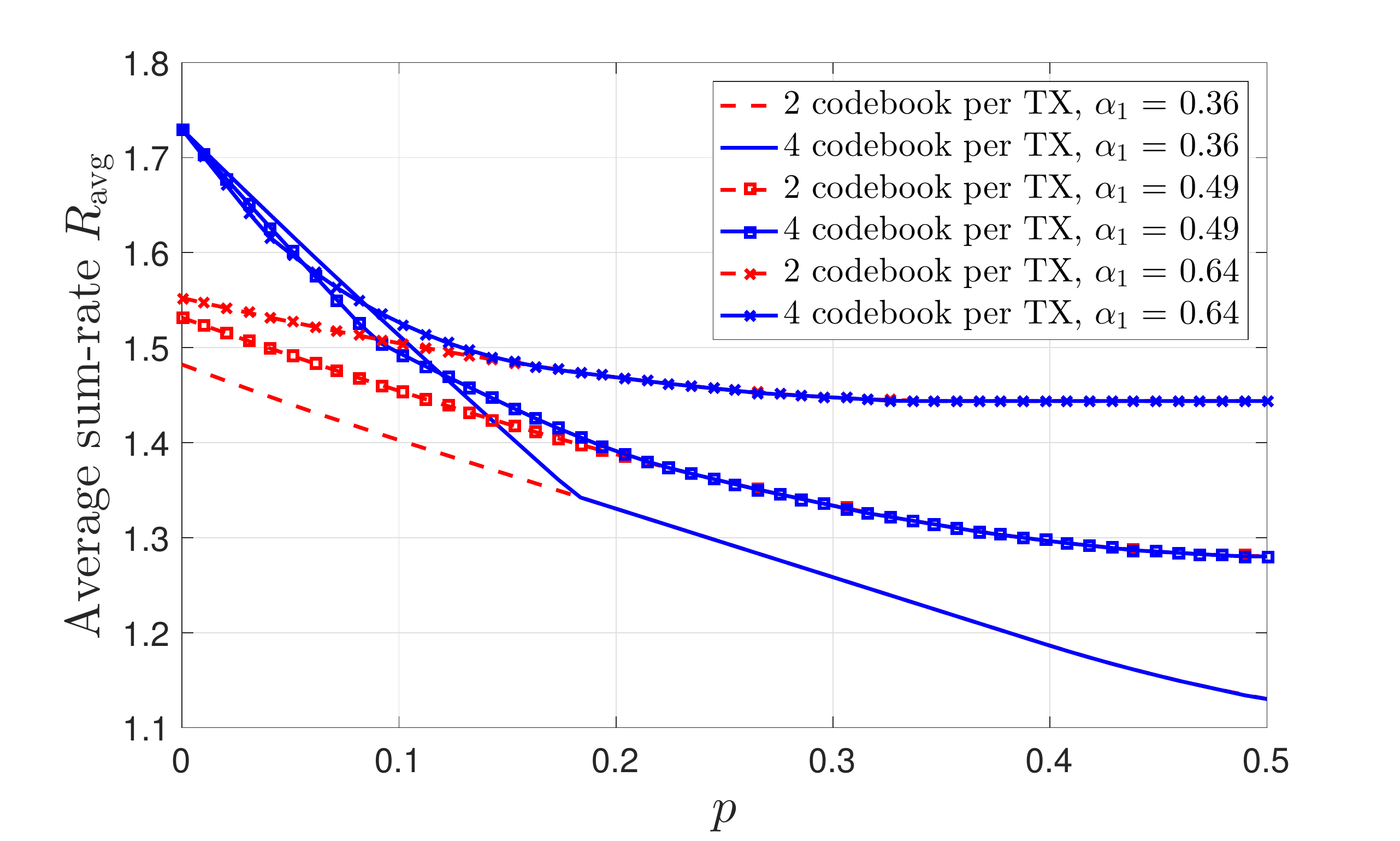}
  \caption{Average sum-rate versus $p$ for different values of $\alpha_1$ ($\alpha_2=1$ and SNR=5).}
  \label{fig:rate_prob}
\end{minipage}
\end{figure}

Next, we evaluate the average rate as a relevant and proper measure for characterizing the performance of the proposed  approach. The average rate is achievable with sufficient number of transmission cycles, where each cycle undergoes an independent fading realization. We consider a symmetric model, in which the corresponding information streams are allocated identical power, and have the same rate, and set $R_{uv} \dff  R^1_{uv}=R^2_{uv}$ for $u,v\in\{1,2\}$. Also, we consider a symmetric distribution for $h_1$ and $h_2$ such that  $\mathbb{P}(h_1^2=\alpha_i)=\mathbb{P}(h_2^2=\alpha_i)$ for $i\in\{1,2\}$, and define $p\dff \mathbb{P}(h_1^2=\alpha_1)=\mathbb{P}(h_2^2=\alpha_1)$. By leveraging the stochastic model of the fading process, the average rate is
\begin{align}\label{eq:Rave}
R_{\rm avg} \dff 2[R_{11} + (1-p)(R_{12}+R_{21})+(1-p)^2R_{22}]\ .
\end{align}
Based on the average rate in \eqref{eq:Rave}, we present the average rate of the proposed approach and compare it with that yielded by the approach of~\cite{Minero:ISIT07}. For the numerical evaluations we consider a two-state channel in which we fix the strong channel by setting $\alpha_2=1$ and let the weak channel $\alpha_1$ vary between 0 and 1. In all settings, we assume that the SNR is 5 dB.  Specifically, the results in Fig.~\ref{fig:rate_alpha} depict the variations of the maximum average rate versus $\alpha_1\in[0.25,1]$ for different choices of the probability $p$.  Based on these notations, Fig.~\ref{fig:rate_alpha} depicts the variations of the maximum average rate versus $\alpha_1$ and for different values of $p$. It is observed that for a wide range of $\alpha_1$ the proposed approach shows considerable gains, and as $p$ (i.e., the probability of encountering a weak channel) decreases, the performance gaps becomes even more significant. Small values of $p$, essentially, capture the settings in which both channels have similar qualities with a high probability. In Fig.~\ref{fig:rate_alpha}, as $\alpha_1$ increases, the average rate initially  decreases, and after reaching its minimum the trend is reversed. This minimum point moves towards the lower values of $\alpha_1$ as  $p$, i.e., the likelihood of encountering a weak channel, increases. The reason underlying this trend is that under interference, the overall quality (e.g., sum-rate) depends on the relative strengths of the direct and interfering links, rather than their absolute values. Since each of the four receivers decodes a number of codebooks from each transmitter, and treat the rest as Gaussian noise, by changing $\alpha_1$ on the one hand the codebooks to be decoded from transmitter 1 enjoy a higher quality channel, and on the other hand, all the remaining codebooks from the same transmitter impose higher interference. Hence, overall, by monotonically changing $\alpha_1$, we cannot expect to observe a monotonic change in the sum-rate. 

\begin{figure}[t]
\centering
\includegraphics[width=3.3in]{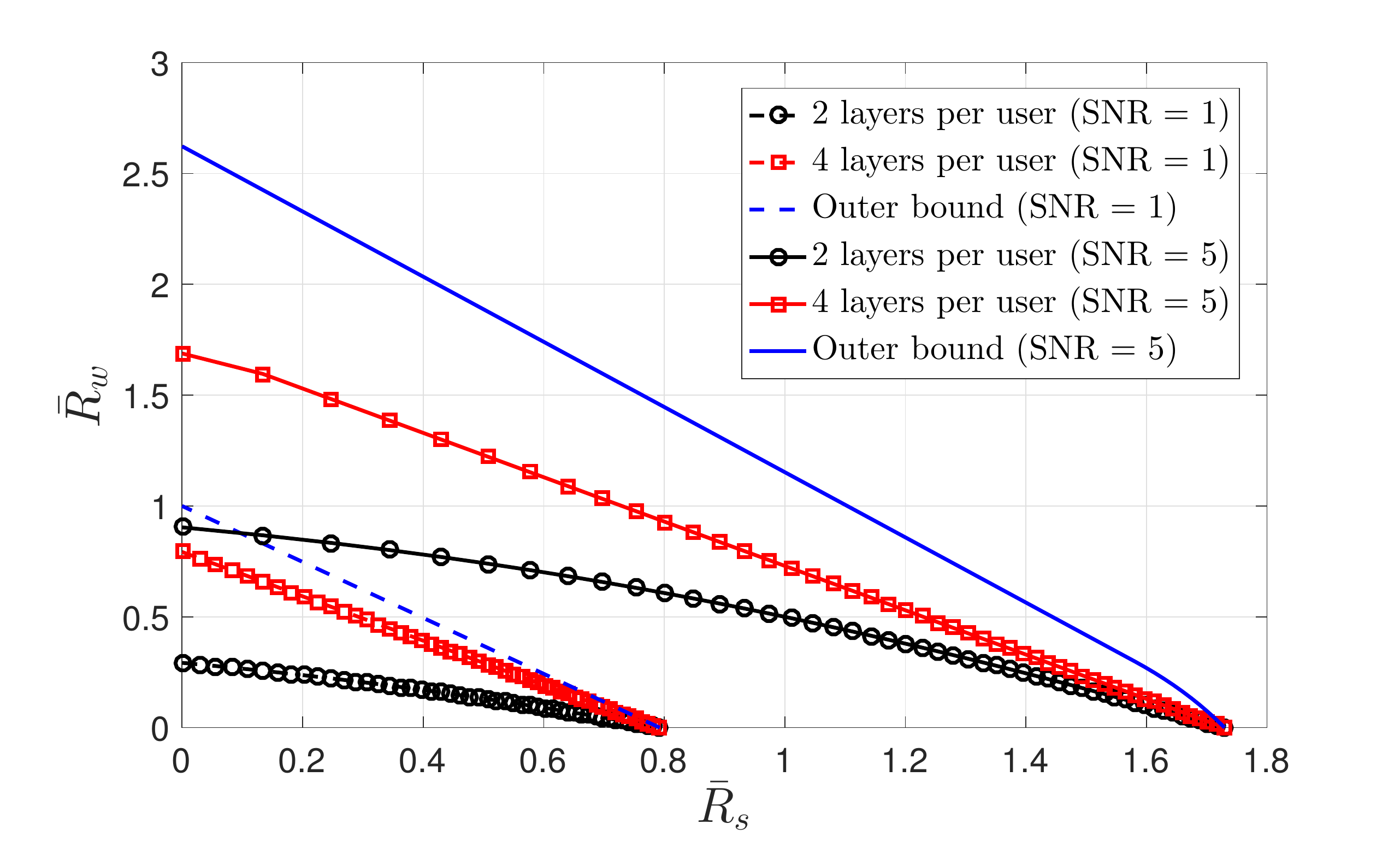}
\caption{Comparison of the capacity region of~\cite{Minero:ISIT07} and regions characterized by theorems~\ref{theorem:achievable_rate2} and~\ref{theorem:outer_bound2}.}
\label{fig:outer_bound2}
\end{figure}

Figure~\ref{fig:rate_prob} depicts the variations of the average sum-rate versus $p$ and for different values of $\alpha_1$. The observations from this figure also confirm that higher gain levels are exhibited as $p$ decreases.
It is noteworthy that the results from Fig.~\ref{fig:region_snr_10} validates the observations from Fig.~\ref{fig:rate_alpha} and Fig.~\ref{fig:rate_prob} that improvement in average rate is significant when probability of encountering weak channel state is low since the rate distribution considered in the achievable rate region comparison will correspond to average rate if probability of observing $\alpha_1$ is zero.

Finally, we assess the relative proximity of the outer bound defined in Theorem~\ref{theorem:outer_bound2} to the achievable rate region presented in Theorem~\ref{theorem:achievable_rate2} for two different levels of SNR. Figure~\ref{fig:outer_bound2} depicts the variations of $\bar R_w$ versus $\bar R_s$ for SNR values   1 and 5, and the choice of $(\sqrt{\alpha_1},\sqrt{\alpha_2}) = (0.5,1)$. Corresponding to each SNR, the figure illustrates the capacity region obtained by the approach of~\cite{Minero:ISIT07}, as well as the achievable rate region and the outer bound of the proposed approach in this paper.  

%However, observed improvement is less significant in Figure~\ref{fig:rate_alpha} and Figure~\ref{fig:rate_prob} since power optimization is applied over all the codebooks for both strategies.
%

\section{Multi-state Channels ($\ell\geq 2$) }
\label{sec:finite}
\vspace{.1 in}
\subsection{Codebook Assignment and Decoding}
In this section, we extend the proposed codebook assignment and decoding strategy designed for the two-state channel to the general multiple-state channel with $\ell\in\mathbb{N}$ states. Similar to the two-state channel, we follow the principle of assigning codebooks based on combined network state, according to which a separate stream of information is designated to each combination of the individual channel states, which necessitates $\ell^2$ codebooks per user. Hence, for $i,j\in\{1,\dots,\ell\}$, the codebook assignment strategy for the users is summarized as follows.

Corresponding to the combined channel state $(h_1^2,h_2^2)=(\alpha_q,\alpha_p)$ we assign codebook $W^1_{pq}$ to User 1 and codebook $W^2_{qp}$ to User 2. By following the same line of analysis as in the two-state channel, the network state $(h_1^2,h_2^2)=(\alpha_1,\alpha_1)$ can be readily verified to be degraded with respect to states $(\alpha_1,\alpha_2)$, $(\alpha_2,\alpha_1)$, and $(\alpha_2,\alpha_2)$ when $\alpha_2>\alpha_1$. Additionally, channel combinations $(\alpha_1,\alpha_2)$ and $(\alpha_2,\alpha_1)$ are also degraded with respect to state $(\alpha_2,\alpha_2)$. When a particular user's channel becomes stronger while the interfering channel remains constant, the user affords to decode additional codebooks. Similarly, when a user's own channel remains constant while the interfering channel becomes stronger, again the user affords to decode additional information. This can be facilitated by decoding and removing the message of the interfering user, based on which the user experiences reduced interference. Based on these observations, for the multiple-state channels we order $h_1$ and $h_2$ in the ascending order and determine their relative degradedness by considering multiple two-state channels with $\alpha_1$ and $\alpha_2$ equal to any two adjacent realizations from the ordered values of $h_i$.  

This strategy is illustrated in Table~\ref{table_3}, in which different channel coefficients $h^2_1$ and $h^2_2$ are listed in the ascending orders. In this table $A_{p,q}$ denotes the cell in the $p^{\rm th}$ row and the $q^{\rm th}$ column, and it specifies the set of codebooks $\mathcal{U}_{pq}$ to be decoded by the combined channel state $(h_1^2,h_2^2)=(\alpha_q,\alpha_p)$. In this table, the set of codebooks to be decoded in each possible combined state is recursively related to the codebooks decoded in the weaker channels. Specifically, the state corresponding to $A_{p-1,q-1}$ is degraded with respect to states $A_{p,q-1}$ and $A_{p-1,q}$. Therefore, in the state $A_{p,q}$, the receiver decodes all streams from states $A_{p-1,q-1}$ (included in ${\cal U}_{p-1,q-1}$), $A_{p,q-1}$ (included in ${\cal U}_{p,q-1}$), and $A_{p-1,q}$ (included in ${\cal U}_{p-1,q}$), as well as one additional stream from each user, i.e., $W^1_{pq}$ and $W^2_{qp}$. 
When both channel coefficients have the highest possible values, all the streams from both users will be decoded at the receiver.

\setlength\extrarowheight{2pt}
\def\arraystretch{1.2}
%\begin{center}
\begin{table*}[t]
\caption{Successive decoding order for the $\ell$-state MAC.}
\label{table_3}
{\scriptsize
\hfill{}
\begin{tabular}{|c||c|c|c|c|c|c|}
\hline
\diagbox{\textcolor{red}{$h^2_2$}}{\textcolor{blue}{$h^2_1$}}& $\alpha_1$ & $\alpha_2$ & $\ldotp\ldotp$ &${\alpha_q}$ & $\ldotp\ldotp$ & $\alpha_\ell$\\
\hline\hline
$\alpha_1$ 
&  \begin{tabular}[t]{@{}c@{}}\\ \textcolor{blue}{$W^1_{11}$} , \textcolor{red}{$W^2_{11}$}\end{tabular} 
& \begin{tabular}[t]{@{}c@{}}$\mathcal{U}_{11}$\\ \textcolor{blue}{$W_{12}^1$} , \textcolor{red}{$W_{21}^2$}\end{tabular}
& $\ldotp\ldotp$
& $\ldotp$
& $\ldotp\ldotp$
& \begin{tabular}[t]{@{}c@{}}$\mathcal{U}_{1(\ell-1)}$\\ \textcolor{blue}{$W_{1\ell}^1$} , \textcolor{red}{$W_{\ell 1}^2$}\end{tabular}\\
\hline
$\alpha_2$ 
& \begin{tabular}[t]{@{}c@{}}$\mathcal{U}_{11}$\\ \textcolor{blue}{$W_{21}^1$} , \textcolor{red}{$W_{12}^2$}\end{tabular}
& \begin{tabular}[t]{@{}c@{}}$\mathcal{U}_{11}$ , $\mathcal{U}_{12}$ , $ \mathcal{U}_{21} $\\\textcolor{blue}{$W_{22}^1$} , \textcolor{red}{$W_{22}^2$}\end{tabular} 
&$\ldotp\ldotp$
&$\ldotp$
&$\ldotp\ldotp$
& \begin{tabular}[t]{@{}c@{}}$\mathcal{U}_{1(\ell-1)} \; , \; \mathcal{U}_{2(\ell-1)} \; , \; \mathcal{U}_{1l}$\\ \textcolor{blue}{$W_{2l}^1$} , \textcolor{red}{$W_{l2}^2$}\end{tabular}\\ 
\hline 
$\ldotp$&$\ldotp$&$\ldotp$&$\ldotp\ldotp$&$\ldotp$ & $\ldotp\ldotp$&$\ldotp$\\
\hline
$\alpha_{p}$ 
& $\ldotp$
& $\ldotp$
& $\ldotp\ldotp$
& \begin{tabular}[t]{@{}c@{}}$\mathcal{U}_{(p-1)(q-1)} , \mathcal{U}_{p(q-1)} , \mathcal{U}_{(p-1)q},$\\ \textcolor{blue}{$W_{pq}^1$} , \textcolor{red}{$W_{qp}^2$}\end{tabular}
& $\ldotp\ldotp$
& $\ldotp$\\
\hline 
$\ldotp$&$\ldotp$&$\ldotp$&$\ldotp\ldotp$&$\ldotp$ &$\ldotp\ldotp$&$\ldotp$\\
\hline
$\alpha_{\ell}$ 
& \begin{tabular}[t]{@{}c@{}}$\mathcal{U}_{(\ell-1)1}$ \\ \textcolor{blue}{$W_{\ell 1}^1$} , \textcolor{red}{$W_{1 \ell}^2$}\end{tabular} 
& \begin{tabular}[t]{@{}c@{}}$\mathcal{U}_{(\ell-1)1},\mathcal{U}_{\ell 1},\mathcal{U}_{(\ell-1)2},$\\\textcolor{blue}{$W_{\ell 2}^1$},\textcolor{red}{$W_{2 \ell}^2$}\end{tabular} 
& $\ldotp\ldotp$
& $\ldotp$
& $\ldotp\ldotp$ 
& \begin{tabular}[t]{@{}c@{}}$\mathcal{U}_{(\ell-1)(\ell-1)}\; , \;\mathcal{U}_{\ell(\ell-1)}\; , \;\mathcal{U}_{(\ell-1)\ell}$\\ \textcolor{blue}{$W_{\ell\ell}^1$} , \textcolor{red}{$W_{\ell\ell}^2$}\end{tabular}\\ 
\hline
\end{tabular}}
\hfill{}
\end{table*}

\subsection{Achievable Rate Region}
In this section, we extend the achievable rate region characterized by Theorem~\ref{theorem:achievable_rate2} for the general multi-state channel. It can be verified that the region characterized by Theorem~\ref{theorem:achievable_rate2} is subsumed by this general rate region as formalized in Corollary~\ref{corollary:2} and shown in Appendix~\ref{app:cor:2}. Similarly to the two-state channel settings, we define $R^i_{uv}$ as the rate of codebook $W^i_{uv}$ for $i\in\{1,2\}$ and $u,v\in\{1,\dots,\ell\}$. We also define $\beta_{uv}\in[0,1]$ as the fraction of the power allocated to the codebook $W^i_{uv}$, where $\sum_{u=1}^\ell\sum_{v=1}^\ell\beta_{uv}=1$. Similarly to the two-state channel setting, for the convenience in notations and for emphasizing the interplay among the rates, we consider a symmetric case in which $R_{uv}\dff R^1_{uv}=R^2_{uv}$.
\begin{theorem}[Achievable Region]
\label{theorem:achievable_rate_finite}
A region of simultaneously achievable rates $$\{R_{uv}: u< v \;\; \mbox{and} \;\; u,v \in\{1,\dots,\ell\}\}$$ for an $\ell$-state two-user multiple access channel is characterized as the set of all rates satisfying:
\begin{align}
\label{R:uv} R_{uv} & \; \leq \; r^1_{uv} \dff \min \left\{b_1(u,v),b_2(u,v),\frac{b_3(u,v)}{2}\right \}\\
\label{R:vu} R_{vu} & \; \leq \;  r^2_{uv} \dff \min \left\{ b_4(u,v),\frac{b_5(u,v)}{2} \right \}\\
\label{R:{vv}} R_{uv}+R_{vu} & \; \leq \;  r^3_{uv} \dff \min \left\{b_6(u,v),b_7(u,v),\frac{b_8(u,v)}{2} \right \}\\
\label{R:5} 2R_{uv}+R_{vu} & \; \leq \; b_9(u,v)\\
\label{R:6} R_{uv}+2R_{vu} & \; \leq \; b_{10}(u,v)\\
\label{R_7} R_{uu} & \leq \min\left\{b_{11}(u),\frac{b_{12}(u)}{2}\right\}\ ,
\end{align}
where constants $\{b_i:i\in\{1,\dots,12\}\}$ are defined in Appendix~\ref{sec:app_b}.
\end{theorem}
%\begin{proof}
%Follows the same footsteps as the proof of  Theorem~\ref{theorem:achievable_rate2}.
%\end{proof}
\begin{corollary}\label{corollary:2}
By setting $\ell=2$, the achievable rate region characterized by Theorem~\ref{theorem:achievable_rate_finite} reduces to the region characterized by Theorem~\ref{theorem:achievable_rate2}.
\end{corollary}
\begin{proof}
See Appendix \ref{app:cor:2}.  
\end{proof}

\section{Conclusions}
\label{sec:conclusions}

We have proposed a broadcast approach for multiple access communication over a slowly fading channel. While the receiver knows the instantaneous channel states, the states are assumed to be unknown to the transmitters. The existing broadcast approaches applied to multiple access communication, directly adopt the approach designed for the single-user channel in which information streams are adapted to the state of the single-user channel. In this paper, we have proposed an encoding strategy in which the information streams are adapted to the combined states of the channels, and have presented a successive decoding strategy for decoding as much information as possible at the receiver, based on the actual channel states. We have characterized the achievable rate region and an outer bound, and have shown that the achievable rate region subsumes the existing known regions in which the information streams are adapted to the single-user channels. 

%\appendices
\appendix

\section{Constants of Theorem~\ref{theorem:achievable_rate2}}
\label{appendix:a}
By defining $\bar\beta_{uv}\dff 1-\beta_{uv}$, the terms used for characterizing the achievable rate region specified in Theorem~\ref{theorem:achievable_rate2} are:
\begin{align}
\label{eq:r11} r_{11}\; & \; \dff \min \Big \{C\big(\alpha_1 \beta_{11},(\alpha_1+\alpha_2)\bar \beta_{11}\big)\; , \;  \frac{1}{2}\; C\big(2\alpha_1 \beta_{11}, 2\alpha_1\bar\beta_{11}\big)\Big\}\ ,\\
\label{eq:r12} r_{12}\; & \; \dff \min \Big \{C\big(\alpha_2 \beta_{12},\alpha_1(\beta_{12}+\beta_{22})+\alpha_2(\beta_{21}+\beta_{22}))\big)\; , \;  \frac{1}{2}\; C\big(2\alpha_2 \beta_{12}, 2\alpha_2\beta_{22}\big)\Big\}\ ,\\
\label{eq:r21} r_{21}\; & \; \dff \min \Big \{C\big(\alpha_1 \beta_{21},\alpha_1(\beta_{12}+\beta_{22})+\alpha_2(\beta_{21}+\beta_{22})\big)\; , \;  \frac{1}{2}\; C\big(2\alpha_2 \beta_{21}, 2\alpha_2\beta_{22}\big)\Big\}\ ,\\
\label{r1_p} r_{1} \; & \; \dff \min \Big \{C\big(\alpha_1\beta_{21}+\alpha_2 \beta_{12},\alpha_1(\beta_{12}+\beta_{22})+\alpha_2(\beta_{21}+\beta_{22})\big)\; , \;  \frac{1}{2}\; C\big(2\alpha_2 (\beta_{12}+\beta_{21}), 2\alpha_2\beta_{22}\big)\Big\}\ ,\\
\label{r12_p} r_{12}' \;  & \; \dff \; C\big ( \alpha_2(2 \beta_{12}+\beta_{21}) \; ,\; 2\alpha_2\beta_{22} \big)\ ,\\
\label{r21_p} r_{21}' \; & \; \dff  \; C\big (\alpha_2(\beta_{12}+2\beta_{21})\; ,\; 2\alpha_2\beta_{22}\big) \ ,\\ 
\label{r22} r_{22} \; & \; \dff \;  \frac{1}{2} C\big (2\alpha_2\beta_{22}\; , \; 0\big)\ .
\end{align}

\section{Proof of Theorem~\ref{theorem:achievable_rate2}}
\label{appendix:theorem:achievable_rate2}
\underline{\textbf{Information Streams $\{W^1_{11},W^2_{11}\}$}}:\\ 
In this section, we first prove that the successive decoding strategy outlined in Table~\ref{table:MAC} for the two-user MAC with two states per channel and no channel state information at the transmitter achieves the region specified in Theorem~\ref{theorem:achievable_rate2}.  Without knowing the CSI, each transmitter sends its message encoded in four separate streams, as specified in Fig.~\ref{fig:code}. At the receiver side, the receiver performs successive decoding by first decoding the streams $W^1_{11}$ and $W^2_{11}$, which are adapted to the {\em weak} channels, i.e., $(h^2_1,h^2_2)=(\alpha_1,\alpha_1)$. At this decoding stage, all other remaining streams are treated as noise. Under such a scheme, successful decoding of these two streams requires that their individual rates and sum rate are within a region characterized by $(R^1_{11},R^2_{11})$ and limited by a set of inequalities that form the boundaries of the capacity region of a two-user MAC. Specifically, streams $W^1_{11}$ and $W^2_{11}$ can be decoded successfully if their corresponding rates satisfy the following conditions under various possible channel state combinations.
\begin{itemize}
\item Channel state $(\alpha_1,\alpha_1)$:
\begin{align}
\label{h1h2:alpha1_alpha1_1}
R^1_{11}&\leq a_1\dff C\left(\alpha_1 \beta^1_{11}\;,\;\alpha_1(\bar\beta^1_{11}+\bar\beta^2_{11})\right)\ , \\
\label{h1h2:alpha1_alpha1_2}R^2_{11}& \leq  a_2\dff  C\left(\alpha_1 \beta^2_{11}\; ,\; \alpha_1 (\bar\beta^1_{11}+\bar\beta^2_{11})\right)\ , \\
\label{h1h2:alpha1_alpha1_3}R^1_{11}+R^2_{11} &\leq  a_3\dff  C\left(\alpha_1(\beta^1_{11}+\beta^2_{11})\; , \; \alpha_1 (\bar\beta^1_{11}+\bar\beta^2_{11})\right) \ ,
\end{align}
\item Channel state $(\alpha_1,\alpha_2)$:
%By repeating this step for other channel state combinations we obtain the following additional set of constraints for  $R^1_{11}$ and $R^2_{11}$.
\begin{align}
\label{h1h2:alpha1_alpha2_1}
R^1_{11}&\leq a_4\dff C\left(\alpha_1 \beta^1_{11}\;,\;\alpha_1\bar\beta^1_{11}+\alpha_2\bar\beta^2_{11}\right)\ , \\
\label{h1h2:alpha1_alpha2_2}R^2_{11}& \leq  a_5\dff  C\left(\alpha_2 \beta^2_{11}\; ,\; \alpha_1 \bar\beta^1_{11}+\alpha_2\bar\beta^2_{11}\right)\ , \\
\label{h1h2:alpha1_alpha2_3}R^1_{11}+R^2_{11} &\leq  a_6\dff  C\left(\alpha_1\beta^1_{11}+\alpha_2\beta^2_{11}\; , \; \alpha_1 \bar\beta^1_{11}+\alpha_2\bar\beta^2_{11}\right) \ .
\end{align}
\item Channel state $(\alpha_2,\alpha_1)$:
\begin{align}
\label{h1h2:alpha2_alpha1_1}  R^1_{11}&\leq a_7\dff C\left(\alpha_2 \beta^1_{11}\;,\;\alpha_2\bar\beta^1_{11}+\alpha_1\bar\beta^2_{11}\right)\ , \\
\label{h1h2:alpha2_alpha1_2} R^2_{11}& \leq  a_8\dff  C\left(\alpha_1 \beta^2_{11}\; ,\; \alpha_2 \bar\beta^1_{11}+\alpha_1\bar\beta^2_{11}\right)\ , \\
\label{h1h2:alpha2_alpha1_3} R^1_{11}+R^2_{11} &\leq  a_9\dff  C\left(\alpha_2\beta^1_{11}+\alpha_1\beta^2_{11}\; , \; \alpha_2 \bar\beta^1_{11}+\alpha_1\bar\beta^2_{11}\right) \ .
\end{align}
\item Channel state $(\alpha_2,\alpha_2)$:
\begin{align}
\label{h1h2:alpha2_alpha2_1}  R^1_{11}&\leq a_{10}\dff C\left(\alpha_2 \beta^1_{11}\;,\;\alpha_2(\bar\beta^1_{11}+\bar\beta^2_{11})\right)\ , \\
\label{h1h2:alpha2_alpha2_2} R^2_{11}& \leq  a_{11}\dff  C\left(\alpha_2 \beta^2_{11}\; ,\; \alpha_2 (\bar\beta^1_{11}+\bar\beta^2_{11})\right)\ , \\
\label{h1h2:alpha2_alpha2_3}R^1_{11}+R^2_{11} &\leq  a_{12}\dff  C\left(\alpha_2(\beta^1_{11}+\beta^2_{11})\; , \; \alpha_2( \bar\beta^1_{11}+\bar\beta^2_{11})\right) \ .
\end{align}
\end{itemize}
From the inequalities in \eqref{h1h2:alpha1_alpha1_1}-\eqref{h1h2:alpha2_alpha2_3}, by comparing all relevant bounds and invoking that $\alpha_1<\alpha_2$, we find the following bounds on $R^1_{11}$, $R^2_{11}$, and $R^1_{11}+R^2_{11}$.
\begin{align}
\label{eq:R11_bound1} R^1_{11}\;&\; \leq \min\{a_1,a_4,a_7,a_{10}\} = a_4\ ,\\
\label{eq:R11_bound2}  R^2_{11}\;&\; \leq \min\{a_2,a_5,a_8,a_{11}\} = a_8 \ ,\\
\label{eq:R11_bound3} R^1_{11}+R^2_{11}\;&\; \leq \min\{a_3,a_6,a_9,a_{12},a_4+a_8\} = \min\{a_3,a_6,a_9,a_4+a_8\}\ .
\end{align}
Furthermore, since the objective is to find the achievable rate region when corresponding streams of the two users have equal rates, and consequently, equal power, we set 
\begin{align}\label{eq:equal}
\beta_{uv}\dff \beta^1_{uv}=\beta^2_{uv} \quad\mbox{and}\quad R_{uv}\dff R^1_{uv}=R^2_{uv}\ .
\end{align}
Based on this assumption, we find that $a_4=a_8$. As a result, the inequalities in \eqref{eq:R11_bound1}-\eqref{eq:R11_bound3} reduce to
\begin{align}
\label{eq:R11_bound4}
R_{11} \leq \min \Big\{a_4,a_8,\frac{a_3}{2},\frac{a_4+a_8}{2}\Big\}=\min \Big \{a_4,\frac{1}{2}a_3\Big\} \overset{\eqref{eq:r11}}{=} r_{11}\ ,
\end{align}
which is the first constraint of the achievable rate region specified in \eqref{R:11}. \vspace{.1 in}

\noindent \underline{\textbf{Information Streams $\{W^1_{12},W^1_{21},W^2_{12},W^2_{21}\}$}:}\\
Next we consider the setting in which one of the two channels is strong. Without loss of generality, assume that $(h_1^2,h_2^2)=(\alpha_2,\alpha_1)$. In such a setting, the streams $W^1_{11}$ and $W^2_{11}$ are already decoded in the first stage of successive decoding at the receiver, and in the second stage, streams $W^1_{12}$ and $W^2_{21}$ will be jointly decoded. In the meantime, streams $\{W^1_{21},W^1_{22}\}$ from user 1, and streams $\{W^2_{12},W^2_{22}\}$ from user 2 are treated as noise. Successful decoding of these information streams is possible if the rates of these streams are within the capacity region of an equivalent MAC transmitting information streams $W^1_{12}$ and $W^2_{21}$ by user 1 and user 2, respectively, while treating other streams as noise. Hence, by following the same line of argument as in the case for weak channels, for various possible states in which streams $\{W^1_{12},W^1_{21},W^2_{12},W^2_{21}\}$ should be decoded, we obtain the following conditions.
\begin{itemize}
\item Channel state $(\alpha_1,\alpha_2)$: In the second stage, information streams $\{W^1_{21},W^2_{12}\}$ are decoded.
\begin{align}
\label{h1h2:alpha1_alpha2_21}
R^1_{21}&\leq a_{13} \dff C\left(\alpha_1 \beta^1_{21}\;,\;\alpha_1(\beta^1_{12}+\beta^1_{22})+\alpha_2(\beta^2_{21}+\beta^2_{22})\right)\ ,\\
\label{h1h2:alpha1_alpha2_12}
R^2_{12}&\leq a_{14}\dff C\left(\alpha_2 \beta^2_{12},\alpha_1 (\beta^1_{12}+\beta^1_{22})+\alpha_2(\beta^2_{21}+\beta^2_{22})\right) \ .\\
\label{h1h2:alpha1_alpha2_12_21}
R^1_{21}+R^2_{12}& \leq a_{15}\dff C\left(\alpha_1\beta^1_{21}+\alpha_2 \beta^2_{12},\alpha_1 (\beta^1_{12}+\beta^1_{22})+\alpha_2(\beta^2_{21}+\beta^2_{22})\right)\ .
\end{align}
\item Channel state $(\alpha_2,\alpha_1)$: In the second stage, information streams $\{W^1_{12},W^2_{21}\}$ are decoded.
\begin{align}
\label{h1h2:alpha2_alpha1_12}
R^1_{12}&\leq a_{16}\dff C\left(\alpha_2 \beta^1_{12},\alpha_2 (\beta^1_{21}+\beta^1_{22})+\alpha_1(\beta^2_{12}+\beta^2_{22})\right)\ ,\\
\label{h1h2:alpha2_alpha1_21}
R^2_{21}& \leq a_{17} \dff C\left(\alpha_1 \beta^2_{21},\alpha_2 (\beta^1_{21}+\beta^1_{22})+\alpha_1(\beta^2_{12}+\beta^2_{22})\right) \ ,\\
\label{h1h2:alpha2_alpha1_12_21}
R^1_{12}+R^2_{21}& \leq a_{18}\dff  C\left(\alpha_2 \beta^1_{12}+\alpha_1 \beta^2_{21},\alpha_2(\beta^1_{21}+\beta^1_{22})+\alpha_1(\beta^2_{12}+\beta^2_{22})\right)\ .
\end{align}
\item  Channel state $(\alpha_2,\alpha_2)$:
In the second stage, information streams  $\{W^1_{12},W^2_{12},W^1_{21},W^2_{21}\}$ are jointly decoded.  Based on this, we obtain the following set of constraints on the rates associated with these information streams. 
\begin{align}
\label{R'_1_12} R^1_{12} \;&\;  \leq  a_{19}\dff C\left(  {\alpha_2 \beta^1_{12} \; , \;\alpha_2 \beta^1_{22} +\alpha_2 \beta^2_{22}  } \right)\ , \\
\label{R'_1_21} R^1_{21} \;&\; \leq  a_{20}  \dff C\left(  {\alpha_2 \beta^1_{21} \; , \;\alpha_2 \beta^1_{22} +\alpha_2 \beta^2_{22}  } \right)\ , \\
\label{R'_2_12} R^2_{12} \;&\; \leq  a_{21}  \dff C\left(  {\alpha_2 \beta^2_{12} \; , \;\alpha_2 \beta^1_{22} +\alpha_2 \beta^2_{22}  } \right)\ ,\\
\label{R'_2_21} R^2_{21} \;&\; \leq  a_{22}  \dff C\left(  {\alpha_2 \beta^2_{21} \; , \;\alpha_2 \beta^1_{22} +\alpha_2 \beta^2_{22}  } \right)\ ,\\
\label{R'_1_12_1_21} R^1_{12}+R^1_{21} \;&\; \leq  a_{23}  \dff C\left(  {\alpha_2 \beta^1_{12} +\alpha_2 \beta^1_{21} \; , \;\alpha_2 \beta^1_{22} +\alpha_2 \beta^2_{22}  } \right)\ ,\\
\label{R'_1_12_2_12} R^1_{12}+R^2_{12} \;&\; \leq  a_{24}  \dff C\left(  {\alpha_2 \beta^1_{12} +\alpha_2 \beta^2_{12} \; , \;\alpha_2 \beta^1_{22} +\alpha_2 \beta^2_{22}  } \right)\ ,\\
\label{R'_1_12_2_21} R^1_{12}+R^2_{21} \;&\; \leq  a_{25}  \dff C\left(  {\alpha_2 \beta^1_{12} +\alpha_2 \beta^2_{21} \; , \;\alpha_2 \beta^1_{22} +\alpha_2 \beta^2_{22}  } \right)\ ,\\
\label{R'_1_21_2_12} R^1_{21}+R^2_{12} \;&\; \leq  a_{26}  \dff C\left(  {\alpha_2 \beta^1_{21} +\alpha_2 \beta^2_{12} \; , \;\alpha_2 \beta^1_{22} +\alpha_2 \beta^2_{22}  } \right)\ ,\\
\label{R'_1_21_2_21} R^1_{21}+R^2_{21} \;&\; \leq  a_{27}  \dff C\left(  {\alpha_2 \beta^1_{21} +\alpha_2 \beta^2_{21} \; , \;\alpha_2 \beta^1_{22} +\alpha_2 \beta^2_{22}  } \right)\ ,\\
\label{R'_2_12_2_21} R^2_{12}+R^2_{21} \;&\; \leq  a_{28}  \dff C\left(  {\alpha_2 \beta^2_{12} +\alpha_2 \beta^2_{21} \; , \;\alpha_2 \beta^1_{22} +\alpha_2 \beta^2_{22}  } \right)\ ,\\
\label{R'_1_12_1_21_2_12} R^1_{12}+R^1_{21}+R^2_{12} \;&\; \leq  a_{29}  \dff C\left(  {\alpha_2 (\beta^1_{12}+ \beta^1_{21}) +\alpha_2 \beta^2_{12} \; , \;\alpha_2 \beta^1_{22} +\alpha_2 \beta^2_{22}  } \right)\ ,\\
\label{R'_1_12_1_21_2_21} R^1_{12}+R^1_{21}+R^2_{21} \;&\; \leq  a_{30}  \dff C\left(  {\alpha_2 (\beta^1_{12}+ \beta^1_{21}) +\alpha_2 \beta^2_{21} \; , \;\alpha_2 \beta^1_{22} +\alpha_2 \beta^2_{22}  } \right)\ ,\\
\label{R'_1_12_2_12_2_21} R^1_{12}+R^2_{12}+R^2_{21} \;&\; \leq  a_{31}  \dff C\left(  {\alpha_2 \beta^1_{12} +\alpha_2 (\beta^2_{12}+\beta^2_{21}) \; , \;\alpha_2 \beta^1_{22} +\alpha_2 \beta^2_{22}  } \right)\ ,\\
\label{R'_1_21_2_12_2_21} R^1_{21}+R^2_{12}+R^2_{21} \;&\; \leq  a_{32}  \dff C\left(  {\alpha_2 \beta^1_{21} +\alpha_2 (\beta^2_{12}+\beta^2_{21}) \; , \;\alpha_2 \beta^1_{22} +\alpha_2 \beta^2_{22}  } \right)\ ,\\
\label{R'_1_12_1_21_2_12_2_21} R^1_{12}+R^1_{21}+R^2_{12}+R^2_{21} \;&\; \leq  a_{33}  \dff C\left(  {\alpha_2 (\beta^1_{12}+\beta^1_{21}) +\alpha_2 (\beta^2_{12}+\beta^2_{21}) \; , \;\alpha_2 \beta^1_{22} +\alpha_2 \beta^2_{22}  } \right)\ .
\end{align}
For the simplicity in notations, and in line with the desired achievable rate region, we assume that the corresponding information streams of the two users have identical rates and powers, as specified in \eqref{eq:equal}. Hence, based on \eqref{h1h2:alpha1_alpha2_12}, \eqref{h1h2:alpha2_alpha1_12}, \eqref{R'_1_12}, \eqref{R'_2_12}, and \eqref{R'_1_12_2_12} for $R_{12}$ it can be easily verified that $a_{14}=a_{16}$ and $a_{19}=a_{21}\geq \frac{1}{2}a_{24}$. Hence, we obtain
\begin{align}\label{R12}
R_{12} \; \leq \; \min\left\{ a_{14},a_{16}, a_{19}, a_{21}, \frac{1}{2} a_{24}\right\} =\min\left\{ a_{14}, \frac{1}{2} a_{24}\right\} \overset{\eqref{eq:r12}}{=}r_{12}\ ,
\end{align}
which is the constraint specified in \eqref{R:12}. Similarly, based on \eqref{h1h2:alpha1_alpha2_21}, \eqref{h1h2:alpha2_alpha1_21}, \eqref{R'_1_21}, \eqref{R'_2_21}, and \eqref{R'_1_21_2_21}, and by leveraging that $a_{13}=a_{17}$ and $a_{20}=a_{22}\geq \frac{1}{2}a_{27}$, for $R_{21}$ we obtain
\begin{align}\label{R21}
R_{21} \; \leq \; \min\left\{ a_{13},a_{17}, a_{20}, a_{22}, \frac{1}{2} a_{27}\right\} =\min\left\{ a_{13}, \frac{1}{2} a_{27}\right\} \overset{\eqref{eq:r21}}{=}r_{21}\ ,
\end{align}
which is the constraint specified in \eqref{R:21}. Next, for obtaining the bound on the sum-rate $(R_{12}+R_{21})$, we leverage~\eqref{h1h2:alpha1_alpha2_12_21}, \eqref{h1h2:alpha2_alpha1_12_21}, \eqref{R'_1_12_1_21}, \eqref{R'_1_12_2_21}, \eqref{R'_1_21_2_12}, \eqref{R'_2_12_2_21}, and \eqref{R'_1_12_1_21_2_12_2_21}, and obtain
\begin{align}
R_{12}+R_{21}\;&\;\leq \min\left\{a_{15},a_{18},a_{23},a_{25}, a_{26},  a_{28}, \frac{1}{2} a_{33}\right\}=\min\left\{a_{15},\frac{1}{2} a_{33}\right\} \overset{\eqref{r1_p}}{=} r_1\ ,
\end{align}
which follows the observation that $a_{15}=a_{18}$ and $a_{23}=a_{25}=a_{26}= a_{28}\geq \frac{1}{2}a_{33}$. By further invoking \eqref{R12} and \eqref{R21} we obtain
\begin{align}
R_{12}+R_{21}\;&\;\leq \min\left\{r_{12}+r_{21}, r_1\right\}\ ,
\end{align}
which after dropping the redundant term simplifies to \eqref{R:1221}.  Next, based on \eqref{R'_1_12_1_21_2_12} and \eqref{R'_1_12_2_12_2_21} we have $a_{29}=a_{31}$, and subsequently,
\begin{align}
2R_{12}+R_{21}\;\leq \; a_{29}\overset{\eqref{r12_p}}{=} r_{12}'\ .
\end{align} 
By further taking into account the constraints on the individual rates $R_{12}$ and $R_{21}$, as well as the constraint on $(R_{12}+R_{21})$, we get
\begin{align}
2R_{12}+R_{21}\;\leq \; \min\{2r_{12}+r_{21} ,r_{12}+r_1 , r_{12}'\}\ ,
\end{align} 
which after dropping the redundant terms, we obtain the desired constraint in \eqref{R:21221}. 
Finally, based on \eqref{R'_1_12_1_21_2_21} and \eqref{R'_1_21_2_12_2_21} we have $a_{30}=a_{32}$, and subsequently,
\begin{align}
R_{12}+2R_{21}\;\leq \; a_{30}\overset{\eqref{r21_p}}{=} r_{21}'\ .
\end{align} 
By further taking into account the constraints on the individual rates $R_{12}$ and $R_{21}$, as well as the constraint on $(R_{12}+R_{21})$, we get
\begin{align}
R_{12}+2R_{21}\;\leq \; \min\{r_{12}+2r_{21} ,r_{21}+r_1 , r_{21}'\}\ ,
\end{align}
which leads to \eqref{R:12221}.  
\end{itemize}

\vspace{.1 in}

\noindent \underline{\textbf{Information Streams $\{W^1_{22}, W^2_{22}\}$}:}\\
Information streams $\{W^1_{22}, W^2_{22}\}$ are jointly decoded only when both channels are strong, i.e.,  $(h_1^2,h_2^2)=(\alpha_2,\alpha_2)$. In this channel state, these two information streams are decoded after the rest are successfully decoded and removed. Hence, all the rates $R_{22}^1$ and $R_{22}^2$ that belong to a MAC consisting of two transmitters with information streams $\{W^1_{22}, W^2_{22}\}$ can be achieved simultaneously. This region is
\begin{align}
R_{22}^1 & \; \leq \; C(\alpha_2\beta_{22}^1 \; , \; 0)\ ,\\
R_{22}^2 & \; \leq \; C(\alpha_2\beta_{22}^2\; , \;0)\ ,\\
R_{22}^1+ R_{22}^2 & \; \leq \; C(\alpha_2(\beta_{22}^1+\beta_{22}^2)\; , \;0)\ .
\end{align}
Hence, under equal power allocation and equal rates in corresponding information streams, we have 
\begin{align}
R_{22}\;&\; \leq \frac{1}{2} C\left(2\alpha_2 \beta_{22}\; , \;0\right)  \overset{\eqref{r22}}{=} r_{22}\ ,
\end{align}
which establishes the constraint in \eqref{R:22}.
 
\section{Proof of Theorem~\ref{theorem:outer_bound2}}
\label{app:outer}
In this section, we derive an upper bound for the capacity region of the network corresponding to Theorem~\ref{theorem:achievable_rate2}. This region is derived by demonstrating that the rates outside this region cannot be achieved with arbitrarily small error rate. Achievable rate region presented in Theorem~\ref{theorem:achievable_rate2}, may or may not coincide with this outer bound depending on values of the power allocation parameters, channel coefficients, and  probability distribution function of the codebooks.

Consider $n$ channel uses, and consequently, codewords with length $n$. Define $\mathcal{W}^1_{ij}\dff\{1,\dots, 2^{nR^1_{ij}}\}$ and $\mathcal{W}^2_{ij}\dff\{1,\dots, 2^{nR^2_{ij}}\}$ as the set of indices of the messages in the information streams $W^1_{ij}$ and  $W^2_{ij}$, respectively. $M^1_{ij}$ and $M^2_{ij}$ are the inputs to the encoders drawn independently and uniformly from the set of messages ${\cal W}^1_{ij}$ and ${\cal W}^2_{ij}$, respectively. For $\ell\in\{1,2\}$ and $\forall M^{\ell}_{ij}$, define $X^{\ell n}_{ij}\dff\mathbf{X}^{\ell n}_{ij}(M^{\ell}_{ij})$ as the output of the encoder of user $\ell$. Similarly, define $(\hat M^1_{ij}(Y^n),\hat M^2_{ij}(Y^n))$ as the output of the decoder. Also, we define $\mathcal{X}^{n}\dff\{X^{1 n}_{11}, X^{1 n}_{12}, X^{1 n}_{21}, X^{1 n}_{22}, X^{2 n}_{11}, X^{2 n}_{12}, X^{2 n}_{21}, X^{2 n}_{22}\}$ as the set of all encoder outputs corresponding to both users.

\noindent \underline{\textbf{Information Streams $\{W^1_{22}, W^2_{22}\}$}:}\\
To determine an upper bound on the rates of $W^1_{22}$ and $W^2_{22}$, we can consider channel state $(\alpha_2,\alpha_2)$ since this is the only channel condition where these two codebooks are decoded. We denote the average error probability by
\begin{align}
{\sf P}_n\dff \mathbb{P}\Big((\hat M^1_{22},\hat M^2_{22})\neq (M^1_{22},M^2_{22})\Big)\ .
\end{align}
By Fano's inequality, conditional entropy of $(M^1_{22},M^2_{22})$ given $Y^n_{22}$ can be expressed as
\begin{align}
\label{equation:Fano}
H(M^1_{22},M^2_{22}|Y^n_{22})&\leq n(R^1_{22}+R^2_{22}){\sf P}_n+H({\sf P}_n)\dff n\epsilon_{22,n}\ ,
\end{align}
where ${\sf P}_n\rightarrow 0$, and subsequently, $\epsilon_{22,n}\rightarrow 0$, as $n\rightarrow \infty$. Hence,
\begin{align}
& n(R^1_{22}+R^2_{22}) 
\\&=H(M^1_{22},M^2_{22})
\\&=I(M^1_{22},M^2_{22};Y^n_{22})+H(M^1_{22},M^2_{22}|Y^n_{22})\\
\label{eq:1} & \leq I(M^1_{22},M^2_{22};Y^n_{22})+n\epsilon_{22,n}\\
\label{eq:2} & \leq I(\mathbf{X}^{1n}_{22}(M^1_{22}),\mathbf{X}^{2n}_{22}(M^2_{22});Y^n_{22})+n\epsilon_{22,n}\\
& = I(X^{1n}_{22},X^{2n}_{22};Y^n_{22})+n\epsilon_{22,n}\\
& \leq I(X^{1n}_{22},X^{2n}_{22};Y^n_{22} \; |\; \mathcal{X}^{n}\backslash \{X^{1n}_{22}, X^{2n}_{22}\})+n\epsilon_{22,n}\\
& = h(Y^n_{22}\;|\;\mathcal{X}^{n}\backslash \{X^{1n}_{22}, X^{2n}_{22}\})-h(Y^n_{22}|\mathcal{X}^{n})+n\epsilon_{22,n}\\
\label{eq:3} &  = \sum\limits_{i=1}^n h(Y_{22,i}|\mathcal{X}_{i}\backslash \{X^{1n}_{22}, X^{2n}_{22}\})-\sum\limits_{i=1}^n h(Y_{22,i}| \mathcal{X}_{i})+n\epsilon_{22,n}\\
&  = \sum\limits_{i=1}^n h(\sqrt{\alpha_2}X^{1}_{22,i} + \sqrt{\alpha_2} X^{2}_{22,i} + N_{22})-\sum\limits_{i=1}^n h(N_{22})+n\epsilon_{22,n}
%\label{eq:4} & =\frac{n}{2}\log\Big(2\pi e(2\alpha_2\gamma P+1)\Big)-\frac{n}{2}\log(2\pi e)+n\epsilon_{22,n} \;
%\\&=\frac{n}{2}\log\Big(\alpha_2(\beta^1_{22}+\beta^2_{22})P+1\Big)+n\epsilon_{22,n}\ ,
\end{align}
where, $\mathcal{X}_{i}$ denotes a set consisting of the $i$th component of each element of $\mathcal{X}^{n}$, \eqref{eq:1} follows by applying inequality~\eqref{equation:Fano}; \eqref{eq:2} follows from  data processing inequality; \eqref{eq:3} follows from the chain rule of entropy function and due to the channel being memoryless. On the other hand, noting that 
\begin{align}
\sum\limits_{i=1}^n h(\sqrt{\alpha_2}X^{1}_{22,i}+\sqrt{\alpha_2}X^{2}_{22,i} + N_{22}) & \; \leq \; \frac{n}{2}\log\Big(2\pi e(2\alpha_2P + 1)\Big)
\end{align}
implies that there exists constant $\beta_{22}\in(0,1)$ corresponding to which
\begin{align}
\label{eq:k1b22} \sum\limits_{i=1}^n h(\sqrt{\alpha_2}X^{1}_{22,i}+\sqrt{\alpha_2}X^{2}_{22,i} + N_{22}) & \; = \; \frac{n}{2}\log\Big(2\pi e(2\alpha_2\beta_{22}P + 1)\Big)
\end{align}
By leveraging $\sum\limits_{i=1}^n h(N_{22})=\frac{n}{2}\log(2\pi e)$ we find 
\begin{align}
n(R^1_{22}+R^2_{22}) & \leq \sum\limits_{i=1}^n h(\sqrt{\alpha_2}X^{1}_{22,i} + \sqrt{\alpha_2} X^{2}_{22,i} + N_{22})-\sum\limits_{i=1}^n h(N_{22})+n\epsilon_{22,n}\\
& = \frac{n}{2}\log\Big(2\pi e(2\alpha_2\beta_{22}P + 1)\Big)-\frac{n}{2}\log\Big(2\pi e\Big)\\
& 
= \frac{n}{2}\log\Big(2\alpha_2\beta_{22}P + 1\Big)\ ,
\end{align}
and as a result
\begin{align}
R^1_{22}+R^2_{22} & \leq \frac{1}{2}\log\Big(2\alpha_2\beta_{22}P + 1\Big) =C(2\alpha_2\beta_{22},0)= 2r_{22}\ .
\end{align}

\noindent \underline{\textbf{Information Streams $\{W^1_{11}, W^2_{11}\}$}:}\\
%Information streams $\{W^1_{11}, W^2_{11}\}$ are always decoded regardless of the channel states. We can derive an outer bound by considering the network state with two weak channels which imposes the minimum upper bound on the sum rate over all possible channel state combinations. 
By following the same steps presented for the information streams $(W^1_{22}, W^2_{22})$ in the previous part we have
\begin{align}
n(R^1_{11}+R^2_{11}) & \; \leq \; I(X^{1n}_{11},X^{2n}_{11};Y^n_{11})+n\epsilon_{11,n}\\
& \; = \; h(Y^n_{11})-h(Y^n_{11}|X^{1n}_{11}, X^{2n}_{11})+n\epsilon_{11,n}\\
& \; = \; \sum\limits_{i=1}^n h(Y_{11,i})-\sum\limits_{i=1}^n h(Y_{11,i}| X^{1}_{11,i}, X^{2}_{11,i})+n\epsilon_{11,n}\ .
\end{align}
Next, note that channel
\begin{align}
Y_{11,i}=\sqrt{\alpha_1}X^1_i+\sqrt{\alpha_1}X^2_i + N_{11, i}\ 
\end{align} 
is statistically equivalent to
\begin{align}
\tilde{Y}_{11,i}\dff \sqrt{\alpha_2}X^1_i+\sqrt{\alpha_2}X^2_i + N_{22, i} + \tilde{N}_{11, i}\ ,
\end{align}
where ${\rm var}(\tilde{N}_{11, i})=\frac{\alpha_2}{\alpha_1}-1$, and $\tilde{Y}_{11,i}=Y_{22,i} + \tilde{N}_{11, i}$. Therefore, 
\begin{align}
n(R^1_{11}+R^2_{11}) & \; \leq \; \sum\limits_{i=1}^n h(\tilde{Y}_{11,i})-\sum\limits_{i=1}^n h(\tilde{Y}_{11,i}| X^{1}_{11,i}, X^{2}_{11,i})+n\epsilon_{11,n}\ .
\end{align}
Next, note that
\begin{align}
\sum\limits_{i=1}^n h(\tilde{Y}_{11,i}) & \; \leq \; \frac{n}{2}\log\Big(2\pi e(2\alpha_2P + \frac{\alpha_2}{\alpha_1})\Big)\ ,
\end{align}
and
\begin{align}
\sum\limits_{i=1}^n h(\tilde{Y}_{11,i}| X^{1n}_{11}, X^{2n}_{11}) & \; = \; \sum\limits_{i=1}^n h(Y_{22,i}+\tilde{N}_{11, i}| X^{1}_{11,i}, X^{2}_{11,i})\\
& \; \geq \; \sum\limits_{i=1}^n h(Y_{22,i}+\tilde{N}_{11, i}| X^{1}_{11,i}, X^{2}_{11,i}, X^{1}_{12,i}, X^{2}_{12,i},X^{1}_{21,i}, X^{2}_{21,i}) \\
& \label{eq:103} \; = \; \frac{n}{2}\log\Big(2\pi e(2\alpha_2\beta_{22}P + \frac{\alpha_2}{\alpha_1})\Big)\ .
\end{align}
As a result, there exist $\beta_{11}\in [0,1-\beta_{22}]$ such that
\begin{align}
\sum\limits_{i=1}^n h(\tilde{Y}_{11,i}| X^{1n}_{11}, X^{2n}_{11}) \; = \; \frac{n}{2}\log\Big(2\pi e(2\alpha_2(1-\beta_{11})P + \frac{\alpha_2}{\alpha_1})\Big)\ .
\end{align}
Therefore, 
\begin{align}
n(R^1_{11}+R^2_{11}) & \; \leq \; \frac{n}{2}\log\Big(2\pi e(2\alpha_2P + \frac{\alpha_2}{\alpha_1})\Big) - \frac{n}{2}\log\Big(2\pi e(2\alpha_2(1-\beta_{11})P + \frac{\alpha_2}{\alpha_1})\Big) \ ,
\end{align}
which implies that
\begin{align}
R^1_{11}+R^2_{11} & \; \leq \; \frac{1}{2}\log\Big(1 + \frac{2\alpha_1(1-(1-\beta_{11}))P}{1 + 2\alpha_1(1-\beta_{11})P}\Big) \\
& \; = \; \frac{1}{2}\log\Big(1 + \frac{2\alpha_1\beta_{11}P}{1 + 2\alpha_1(1-\beta_{11})P}\Big) = a_3
%&\\ \; \leq \; \frac{1}{2}\log\Big(1 +2\alpha_1\beta_{11}P\Big)\dff a_{34}\ .
\end{align}

\noindent \underline{\textbf{Information Streams $\{W^1_{12},W^1_{21},W^2_{12},W^2_{21}\}$}:}\\
Next, we determine an outer bound on the rates of information streams $W^i_{12}$ and $W^i_{21}$ for  $i\in\{1,2\}$. For this purpose, we focus the channel state $(h^2_1,h^2_2)=(\alpha_2,\alpha_2)$ and obtaining a set of outer bounds. By following the same line of analysis, it can be readily shown that the constraints enforced by other channel state combinations will be redundant. %This upper bound corresponds to the scenario when $\{W^1_{11}, W^2_{11}, W^1_{21}, W^2_{21}\}$ are not present, i.e., 
\begin{align}
&n(R^1_{12}+R^2_{12}) \\
& \; \leq \; I(X^{1n}_{12},X^{2n}_{12};Y^n_{22})+n\epsilon_{22',n}\\
& \; \leq \; I(X^{1n}_{12},X^{2n}_{12};Y^n_{22} \; |\; \mathcal{X}^{n}\backslash\{X^{1n}_{12}, X^{2n}_{12}, X^{1n}_{22},  X^{2n}_{22}\})+n\epsilon_{22',n}\\
& \; = \; h(Y^n_{22}\;|\;\mathcal{X}^{n}\backslash\{X^{1n}_{12}, X^{2n}_{12}, X^{1n}_{22},  X^{2n}_{22}\}))-h(Y^n_{22}|\mathcal{X}^{n}\backslash\{X^{1n}_{22} , X^{2n}_{22}\})+n\epsilon_{22',n}\\
&  \; = \; \sum\limits_{i=1}^n h(Y_{22,i}|\mathcal{X}_{i}\backslash\{X^{1}_{12,i}, X^{2}_{12,i}, X^{1}_{22,i},  X^{2}_{22,i}\}))-\sum\limits_{i=1}^n h(Y_{22,i}|\mathcal{X}_{i}\backslash\{X^{1}_{22,i}, X^{2}_{22,i}\})+n\epsilon_{22',n}\ .
%&  = \sum\limits_{i=1}^n h(\sqrt{\alpha_2}X^{1}_{22,i} + \sqrt{\alpha_2} X^{2}_{22,i} + N_{22})-\sum\limits_{i=1}^n h(N_{22})+n\epsilon_{22',n}\ .
\end{align}
Now, from ~\eqref{eq:k1b22} we obtain 
\begin{align}
\sum\limits_{i=1}^n h(Y_{22,i}|\mathcal{X}_{i}\backslash\{X^{1}_{22,i}, X^{2}_{22,i}\}) = \frac{n}{2}\log\Big(2\pi e(2\alpha_2\beta_{22}P + 1)\Big)\ .
\end{align}
Also, since 
\begin{align}
h(Y_{22,i}|\mathcal{X}_{i}\backslash\{X^{1}_{22,i}, X^{2}_{22,i}\}) & = h(\sqrt{\alpha_2}X^{1}_{22,i}+\sqrt{\alpha_2}X^{2}_{22,i} + N_{22})\ ,
\end{align}
and
\begin{align}
h(Y_{22,i}|\mathcal{X}_{i}\backslash\{X^{1}_{12,i}, X^{2}_{12,i},X^{1}_{22,i}, X^{2}_{22,i}\}) & = h(\sqrt{\alpha_2}(X^{1}_{12,i} + X^{1}_{22,i})+\sqrt{\alpha_2}(X^{2}_{12,i}+X^{2}_{22,i}) + N_{22})\ ,
\end{align} 
by comparing the variance values of arguments of the two entropy terms we have 
\begin{align}
h(Y_{22,i}|\mathcal{X}_{i}\backslash\{X^{1}_{12,i}, X^{2}_{12,i},X^{1}_{22,i}, X^{2}_{22,i}\}) \geq h(Y_{22,i}|\mathcal{X}_{i}\backslash\{X^{1}_{22,i}, X^{2}_{22,i}\})\ .
\end{align}
Therefore, there exists $\beta_{12}\in [0,1-\beta_{22}]$ such that
\begin{align}
\sum\limits_{i=1}^n h(Y_{22,i}|\mathcal{X}_{i}\backslash\{X^{1}_{12,i}, X^{2}_{12,i},X^{1}_{22,i}, X^{2}_{22,i}\})  = \frac{n}{2}\log\Big(2\pi e(2\alpha_2(\beta_{22}+\beta_{12})P + 1)\Big)\ ,
\end{align}
and subsequently,
\begin{align}
& n(R^1_{12}+R^2_{12}) \\
& \; \leq \; \sum\limits_{i=1}^n h(Y_{22,i}|\mathcal{X}_{i}\backslash\{X^{1}_{12,i}, X^{2}_{12,i},X^{1}_{22,i}, X^{2}_{22,i}\}) -\sum\limits_{i=1}^n h(Y_{22,i}|\mathcal{X}_{i}\backslash\{X^{1}_{22,i}, X^{2}_{22,i}\})+n\epsilon_{22',n} \\
& \; = \; \frac{n}{2}\log\Big(2\pi e(2\alpha_2(\beta_{22}+\beta_{12})P + 1)\Big) - \frac{n}{2}\log\Big(2\pi e(2\alpha_2\beta_{22}P + 1)\Big) + n\epsilon_{22',n} \\
& \; = \; \frac{n}{2}\log\Big( 1 + \frac{2\alpha_2\beta_{12}P}{1 + 2\alpha_2\beta_{22}P}\Big) + n\epsilon_{22',n}\ .
\end{align}
As a result,
\begin{align}
R^1_{12}+R^2_{12} & \; \leq \; \frac{1}{2}\log\Big( 1 + \frac{2\alpha_2\beta_{12}P}{1 + 2\alpha_2\beta_{22}P}\Big) \; = \; a_{24}\ .
\end{align}
Similarly, we can find the following upper bound for information streams $(W^1_{21}, W^2_{21})$:
\begin{align}
 R^1_{21}+R^2_{21} & \; \leq \; \frac{1}{2}\log\Big( 1 + \frac{2\alpha_2\beta_{21}P}{1 + 2\alpha_2\beta_{22}P}\Big) = a_{27}\ ,
\end{align}
which concludes the proof.

\section{Proof of Corollary~\ref{corollary:1}}
\label{appendix:corollary:1}

Similarity in rate splitting, superposition coding, and successive decoding becomes apparent by directly comparing the entries of Table~\ref{table:tse} and Table~\ref{table:MAC} after setting the power allocated to streams $\{W^1_{21},W^2_{21},W^1_{22},W^2_{22}\}$ equal to zero along with a renaming of the information streams $\{W^i_{11},W^i_{12}\}$ to $\{W^i_{1},W^i_{2}\}$ for $i\in\{1,2\}$.

In the next step, we show that for the given power allocation scheme specified in Corollary~\ref{corollary:1}, the achievable rate region characterized in Theorem~\ref{theorem:achievable_rate2} coincides with the capacity region presented in \cite{Minero:ISIT07}. We start by setting the power allocated to the streams $\{W^1_{21},W^2_{21},W^1_{22},W^2_{22}\}$ to zero, i.e., 
\begin{align}
\beta^1_{21}=\beta^2_{21}=\beta^1_{22}=\beta^2_{22}\;=\; 0\ .
\end{align}
Based on this, the part of the achievable rate region characterized in \eqref{eq:R11_bound1}-\eqref{eq:R11_bound3} simplifies to:
\begin{align}
\label{eq:R11_bound1_2} R^1_{11}\;&\; \leq a_4 = C\big(\alpha_1 \beta^1_{11},\alpha_1\beta_{12}^1+\alpha_2\beta^2_{12}\big)\\
\label{eq:R11_bound2_2}  R^2_{11}\;&\; \leq  a_8= C\big(\alpha_1 \beta^2_{11},\alpha_2\beta_{12}^1+\alpha_1\beta^2_{12}\big)\\
\label{eq:R11_bound3_2} R^1_{11}+R^2_{11}\;&\; \leq \min\{a_3,a_6,a_9,a_4+a_8\}\ .
\end{align}
By comparing the capacity region presented in \cite{Minero:ISIT07}, the sum-rate of the two lower information streams $W_{11}^1$ and $W_{11}^2$ will be less than or equal to the minimum of two sum-rates.  One is the minimum of the sum-rate under different combination of channel states when one or both users have weak channels, i.e., $\min\{a_3,a_6,a_9\}$. The second sum-rate constrained is obtained by aggregating the constraints on the individual rates for information streams $W_{11}^1$ and $W_{11}^2$ which allow them to be decodable in all four possible channel states. The individual rate constraints take their smallest values when the interfering channel is strong while the user's own channel coefficient is weak. Furthermore, the survivors of all the constraints in \eqref{R'_1_12}-\eqref{R'_1_12_1_21_2_12_2_21} simplify to:
\begin{align}
\label{R'_1_12_2} R^1_{12} \;&\;  \leq  a_{19}\dff C\left(  {\alpha_2 \beta^1_{12} \; , 0  } \right) \\
\label{R'_2_12_2} R^2_{12} \;&\; \leq  a_{21}  \dff C\left(  {\alpha_2 \beta^2_{12} \; , \; 0  } \right)\\
\label{R'_1_12_2_12_2} R^1_{12}+R^2_{12} \;&\; \leq  a_{24}  \dff C\left(  {\alpha_2 \beta^1_{12} +\alpha_2 \beta^2_{12} \; , \; 0} \right)\ .
\end{align}
The combination of \eqref{eq:R11_bound3_2} and \eqref{R'_1_12_2_12_2} establishes the achievable rate region based on the codebook assignment specified in \cite{Minero:ISIT07} and in  Table~\ref{table:tse}.

\section{Values of $\{b_i:i\in\{1,\dots,12\}\}$}
\label{sec:app_b}
By defining the sets
\begin{align}
\label{eq:J1} J_1(u,v) & \dff \{j\in \{u,\dots, v-1\}\}\ ,\\
\label{eq:J2} J_2(u,v) & \dff \{(j,k): k\in\{u,\dots,v-1\} \; \; \& \;\; j\in\{v+1,\dots,\ell\}\}\ ,\\
\label{eq:J3} J_3(u,v) & \dff \{(j,k): j\leq k  \; \; \& \;\;  j,k\in\{v,\dots,\ell\}\}\ ,
\end{align}
we have
{
\begin{align}
\label{eq:b1} b_1(u,v) & \dff  \quad   \min_{j\in J_1}\left\{ C\big ( {\alpha_v \beta_{uv}\; , \; \alpha_jB_{1}(j,u,v)+\alpha_vB_{2}(j,u,v)}\big) \right \}\ ,\\
\label{eq:b2} b_2(u,v) & \dff C\big ( {\alpha_v \beta_{uv} \; , \; (\alpha_v+\alpha_\ell)B_{3}(u,v)}\big)\ ,\\
\label{eq:b3} b_3(u,v)  & \dff  C\big ( {2\alpha_v \beta_{uv}\; , \; 2\alpha_vB_{3}(u,v)} \big)\ ,\\
\label{eq:b4} b_4(u,v) & \dff C\big ( {\alpha_u \beta_{vu}\; , \;  \alpha_\ell B_{4}(u,v)+\alpha_uB_{5}(u,v)}\big)\ ,\\
\label{eq:b5} b_5(u,v) & \dff     C\big ( 2\alpha_v \beta_{vu} \; , \;  2\alpha_vB_3(u,v) \big)\ ,\\
\label{eq:b6} b_6(u,v)&  \dff  \;\;  \min_{(j,k)\in J_2}\left\{ C\big ( {\alpha_j \beta_{vu}+\alpha_k\beta_{uv}\; , \; \alpha_jB_{6}(k,u,v)+\alpha_kB_{7}(k,u,v)}\big)\right \}\ ,\\
\label{eq:b7} b_7(u,v)&  \dff C\big ( {\alpha_v (\beta_{uv}+\beta_{vu})\; , \; (\alpha_v+\alpha_\ell)B_{3}(u,v)}\big)\ ,\\
\label{eq:b8} b_8(u,v)&  \dff    C\big ( {2\alpha_v (\beta_{uv}+\beta_{vu})\; , \; 2\alpha_vB_{3}(u,v)} \big)\ ,\\
\label{eq:b9} b_9(u,v) & \dff \min_{(j,k)\in J_3} \left\{ C\big (\alpha_j(\beta_{uv}+\beta_{vu})+\alpha_k\beta_{uv}\; , \;  (\alpha_j+\alpha_k)B_{3}(u,v)\big)\right \}\ ,\\
\label{eq:b10} b_{10}(u,v) & \dff \min_{j,k\in J_3}\left\{ C\big ( {\alpha_j(\beta_{uv}+\beta_{vu})+\alpha_k \beta_{vu}\; , \; (\alpha_j+\alpha_k)B_{3}(u,v))}\big )\right \}\ ,\\
\label{eq:b11} b_{11}(u) & \dff C\big ( {\alpha_u\beta_{uu}\; , \; (\alpha_u+\alpha_\ell)B_8(u,u)}\big)\ ,\\
\label{eq:b12} \mbox{and}\qquad b_{12}(u) & \dff C\big ({2\alpha_u\beta_{uu} \; , \; 2\alpha_uB_8(u,u)}\big)\ ,
\end{align}
}
where were have defined
\begin{align}
\label{eq:B1} B_{1}(j,u,v) & \dff 1 - \sum_{n=1}^j\sum_{m=1}^{v-1}\beta_{mn}-\sum_{n=1}^u\beta_{vn}\ ,\\
\label{eq:B2} B_{2}(j,u,v) & \dff 1 - \sum_{n=1}^{v-1}\sum_{m=1}^j\beta_{mn}-\sum_{n=1}^u\beta_{nv}\ ,\\
\label{eq:B3} B_{3}(u,v) & \dff 1 - \sum_{n=1}^{v-1}\sum_{m=1}^{v-1}\beta_{mn}-\sum_{n=1}^u\beta_{vn}-\sum_{n=1}^u\beta_{nv}\ ,\\
\label{eq:B4} B_{4}(u,v) & \dff 1 -  \sum_{n=1}^{v-1}\sum_{m=1}^u\beta_{mn}-\sum_{n=1}^u\beta_{nv} \ ,\\
\label{eq:B5} B_{5}(u,v) & \dff 1 - \sum_{n=1}^{u}\sum_{m=1}^{v-1}\beta_{mn}-\sum_{n=1}^u\beta_{vn}\ ,\\
\label{eq:B6} B_{6}(j,u,v) & \dff 1 -  \sum_{n=1}^j\sum_{m=1}^{v-1}\beta_{mn}-\sum_{n=1}^u\beta_{vn} \ ,\\
\label{eq:B7} B_{7}(j,u,v) &  \dff 1 -  \sum_{n=1}^{j}\sum_{m=1}^{v-1}\beta_{nm}-\sum_{n=1}^u\beta_{nv}\\
\label{eq:B8} \mbox{and}\qquad B_{8}(u,v) &  \dff 1 -  \sum_{n=1}^u\sum_{m=1}^v\beta_{mn}\ .
\end{align}

\section{Proof of Corollary~\ref{corollary:2}}
\label{app:cor:2}
\noindent \underline{\textbf{Information Streams $\{W^1_{11},W^2_{11}\}$}:}\\
From \eqref{R_7} we have
\begin{align}
R_{11} \leq \min\left\{b_{11}(1), \frac{b_{12}(1)}{2}\right\}\ ,
\end{align}
and based on~\eqref{h1h2:alpha1_alpha2_1},~\eqref{eq:b11}, and~\eqref{eq:B8} we find that
\begin{align}
b_{11}(1) = C\Big(\alpha_1\beta_{11}, (\alpha_1+\alpha_2)\bar\beta_{11}\Big)=a_4\ .
\end{align}
Similarly, based on~\eqref{h1h2:alpha1_alpha1_3},~\eqref{eq:b12}, and~\eqref{eq:B8} we have
\begin{align}
b_{12}(1) = C\Big(2\alpha_1\beta_{11}, 2\alpha_1\bar\beta_{11}\Big) = a_3\ .
\end{align}
Therefore, by combining the above two representations for $b_{11}(1)$ and $b_{12}(1)$ we obtain $R_{11}\leq\min\{a_4, \frac{1}{2}a_3\}=r_{11}$ which is the constraint in~\eqref{R:11}\\

\noindent \underline{\textbf{Information Streams $\{W^1_{22},W^1_{22}\}$}:}\\
From~\eqref{R_7} we have
\begin{align}
R_{22} \leq \min\left\{b_{11}(2), \frac{b_{12}(2)}{2}\right\}\ .
\end{align}
By leveraging~\eqref{eq:b11} and~\eqref{eq:B8} we obtain
\begin{align}
b_{11}(2) = C\Big(\alpha_2\beta_{22}, 2\alpha_2(1-\beta_{11}-\beta_{12}-\beta_{21}-\beta_{22})\Big)=C\Big(\alpha_2\beta_{22}, 0\Big)\ .
\end{align}
Similarly, based on~\eqref{eq:b12} and~\eqref{eq:B8} we have
\begin{align}
b_{12}(2) = C\Big(2\alpha_2\beta_{22}, 0\Big)\ .
\end{align}
 Therefore, 
\begin{align}
R_{22} & \leq \min\left\{C\Big(\alpha_2\beta_{22}, 0\Big), \frac{1}{2}C\Big(2\alpha_2\beta_{22}, 0\Big)\right\} = \frac{1}{2}C\Big(2\alpha_2\beta_{22}, 0\Big) = r_{22}\ ,
\end{align}
which is the constraint in~\eqref{R:22}. \\

\noindent \underline{\textbf{Information Streams $\{W^1_{12},W^1_{21},W^2_{12},W^2_{21}\}$}:}\\
From~\eqref{R:uv} we have
\begin{align}
R_{12} \leq \min\left\{b_1(1, 2), b_2(1, 2), \frac{1}{2}b_3(1, 2)\right\}\ .
\end{align}
By leveraging~\eqref{h1h2:alpha1_alpha2_12},~\eqref{eq:b1},~\eqref{eq:B1}, and~\eqref{eq:B2} and we have
\begin{align}
b_{1}(1, 2) & = C\Big(\alpha_2\beta_{12}, \alpha_1(1-\beta_{11}-\beta_{21})+\alpha_2(1-\beta_{11}-\beta_{12})\Big)\\
& = C\Big(\alpha_2\beta_{12}, \alpha_1(\beta_{12}+\beta_{22})+\alpha_2(\beta_{21}+\beta_{22})\Big)\\
& = a_{14}\ .
\end{align}
Furthermore, from~\eqref{eq:b2} and~\eqref{eq:B3} we have
\begin{align}
b_{2}(1, 2) = C\Big(\alpha_2\beta_{12}, 2\alpha_2(1-\beta_{11}-\beta_{12}-\beta_{21})\Big) = C\Big(\alpha_2\beta_{12}, 2\alpha_2\beta_{22}\Big)\ ,
\end{align}
and based on~\eqref{eq:b3},~\eqref{eq:B3}, and~\eqref{R'_1_12_2_12} we have 
\begin{align}
b_{3}(1, 2) & = C\Big(2\alpha_2\beta_{12}, 2\alpha_2(1-\beta_{11}-\beta_{21}-\beta_{12})\Big)\\
& = C\Big(2\alpha_2\beta_{12}, 2\alpha_2\beta_{22}\Big)\\
& = a_{24}\ .
\end{align}
Since $\min\{b_2(1,2),\frac{1}{2}b_3(1,2)\}=\frac{1}{2}b_3(1, 2)$, therefore $R_{12}\leq \min\{a_{14}, \frac{1}{2}a_{24}\}=r_{12}$, which is the constraint in~\eqref{R:12}. Similarly, from~\eqref{R:vu}, we can recover the constraint in~\eqref{R:21} of Theorem~\eqref{theorem:achievable_rate2}.

In order to recover the sum-rate constraint in~\eqref{R:1221}, we set $u=1$ and $v=2$ and based on~\eqref{R:{vv}} we obtain
\begin{align}
R_{12}+R_{21}\leq\min\left\{b_6(1,2), b_7(1,2), \frac{1}{2}b_8(1,2)\right\}\ .
\end{align}
From the definition in~\eqref{eq:J2} we have $J_2=\{(1, 2)\}$. Therefore, by using~\eqref{eq:b6}, ~\eqref{eq:B6}, and~\eqref{eq:B7} we have
 \begin{align}
b_6(1, 2) & = C\Big(\alpha_1\beta_{21}+\alpha_2\beta_{12}, \alpha_1(1-\beta_{11}-\beta_{21})+\alpha_2(1-\beta_{11}-\beta_{12})\Big)\\
& = C\Big(\alpha_1\beta_{21}+\alpha_2\beta_{12}, \alpha_1(\beta_{12}+\beta_{22})+\alpha_2(\beta_{21}+\beta_{22})\Big)\ .
\end{align}
Likewise, from~\eqref{eq:b7}-\eqref{eq:b8} and \eqref{eq:B3} we have 
\begin{align}
b_7(1,2) & = C\Big(\alpha_2(\beta_{12}+\beta_{21}), 2\alpha_2\beta_{22}\Big)\ ,\\
\mbox{and}\qquad b_8(1,2) & = C\Big(2\alpha_2(\beta_{12}+\beta_{21}), 2\alpha_2\beta_{22}\Big)\ .
\end{align}
By noting that $b_7(1,2)\geq \frac{1}{2}b_8(1,2)$ we obtain
\begin{align}
R_{12}+R_{21}\leq\min\{b_6(1,2),\frac{1}{2}b_8(1,2)\}=r_1\ ,
\end{align}
which is the constraint in~\eqref{R:1221}. Finally, by setting $u=1$ and $v=2$ in ~\eqref{R:5}, we find that $2R_{12}+R_{21}\leq b_9(1,2)$. From Equation~\eqref{eq:J3}, $J_3=\{(2,2)\}$, applying which to~\eqref{eq:b9}, and leveraging~\eqref{eq:B3} and~\eqref{r12_p} yields
\begin{align}
b_9(1,2) & = C\Big(\alpha_2(\beta_{12}+\beta_{21})+\alpha_2\beta_{12},2\alpha_2(1-\beta_{11}-\beta_{12}-\beta_{21})\Big) \\
& = C\Big(\alpha_2(2\beta_{12}+\beta_{21}),2\alpha_2\beta_{22}\Big) \\
& = r'_{12}\ ,
\end{align}
which is the constraint in~\eqref{R:21221}. The constraint in~\eqref{R:12221} can be recovered in a similar fashion.

\bibliographystyle{IEEEtran}
\bibliography{TCOM_20171222}

\end{document}